\documentclass[11pt]{article}

\newif\ifnotes
\notesfalse

\title{Binary Error-Correcting Codes with Minimal Noiseless Feedback}

\author{Meghal Gupta\thanks{Email: \texttt{meghal@berkeley.edu}. Supported by a UC Berkeley Chancellor's Fellowship.}\\UC Berkeley 
\and Venkatesan Guruswami\thanks{Email: \texttt{venkatg@berkeley.edu}. Research supported in part by NSF grant CCF-2210823 and a Simons Investigator Award. }\\UC Berkeley
\and Rachel Yun Zhang\thanks{Email: \texttt{rachelyz@mit.edu}. Supported by NSF Graduate Research Fellowship 2141064. Supported in part by DARPA under Agreement No. HR00112020023 and by an NSF grant CNS-2154149.}\\MIT}
\date{\today}

\usepackage{hyperref}
\usepackage[margin=1in]{geometry}
\usepackage{mathptmx}
\usepackage{amssymb}
\usepackage{amsmath}
\usepackage{amsthm}
\usepackage{amsfonts}
\usepackage{bm}
\usepackage{enumitem}
\usepackage{color}
\usepackage{comment}
\usepackage[capitalize]{cleveref}
\usepackage[dvipsnames]{xcolor}
\usepackage{float}
\usepackage[T1]{fontenc}
\usepackage{todonotes}
\usepackage{asymptote}
\usepackage{mdframed}
\usepackage[most]{tcolorbox}
\usepackage{hyperref}
\usepackage{enumitem}
\usepackage{framed}
\usepackage{mdframed}
\usepackage{scrextend}
\usepackage{multirow}
\usepackage{ifthen}
\usepackage{bbm}
\usepackage{frcursive}
\usepackage{thm-restate}

\usepackage[T1]{fontenc}

\definecolor{denim}{rgb}{0.08, 0.38, 0.74}
\definecolor{classicrose}{rgb}{0.98, 0.8, 0.91}
\definecolor{darkpastelblue}{rgb}{0.47, 0.62, 0.8}
\definecolor{dogwoodrose}{rgb}{0.84, 0.09, 0.41}

\usepackage{hyperref}
\hypersetup{
    colorlinks=true,
    linkcolor=dogwoodrose,
    filecolor=dogwoodrose,
    citecolor=denim,
    urlcolor=denim,
}
\usepackage[hyperpageref]{backref}

\newtheorem{theorem}{Theorem}[section]
\newtheorem{lemma}[theorem]{Lemma}
\newtheorem*{lemma*}{Lemma}
\newtheorem{claim}[theorem]{Claim}

\newtheorem{corollary}[theorem]{Corollary}

\newtheorem{definition}[theorem]{Definition}

\theoremstyle{definition}

\Crefname{theorem}{Theorem}{Theorems}
\Crefname{claim}{Claim}{Claims}
\Crefname{lemma}{Lemma}{Lemmas}
\Crefname{proposition}{Proposition}{Propositions}
\Crefname{corollary}{Corollary}{Corollaries}
\Crefname{definition}{Definition}{Definitions}

\newcommand{\ECC}{\mathsf{ECC}}
\newcommand{\ecc}{\mathsf{ecc}}

\newcommand{\code}{\mathsf{FC}}

\newcommand{\update}{\mathsf{update}}

\newcommand{\partition}{\mathsf{ptdesc}}
\newcommand{\partdec}{\mathsf{pt}}

\newcommand{\mini}{\mathtt{mini}}
\newcommand{\main}{\mathtt{main}}
\newcommand{\pair}{\mathtt{pair}}

\renewcommand{\le}{\leqslant}
\renewcommand{\leq}{\leqslant}
\renewcommand{\ge}{\geqslant}
\renewcommand{\geq}{\geqslant}

\newcommand{\nxt}{\mathsf{next}}

\newcommand{\maj}{\text{maj}}

\newcommand{\codeC}{\mathsf{C}}

\newcommand{\poly}{\text{poly}}

\newcommand{\pos}{\mathsf{pos}}
\newcommand{\posx}{\mathsf{posx}}

\newcommand{\bbN}{\mathbb{N}}

\newcommand{\bbR}{\mathbb{R}}
\newcommand{\bbZ}{\mathbb{Z}}

\newcommand{\cG}{\mathcal{G}}

\newcommand{\cL}{\mathcal{L}}

\newcommand{\cP}{\mathcal{P}}

\makeatletter
\newcommand{\customlabel}[2]{%
   \protected@write \@auxout {}{\string \newlabel {#1}{{#2}{\thepage}{#2}{#1}{}} }%
   \hypertarget{#1}{#2}
}
\makeatother

\newcommand{\protocol}[3]{
    \stepcounter{figure}
    \vspace{0.15cm}
    { \small
    \begin{tcolorbox}[breakable, enhanced, colback=classicrose!20]
    \begin{center}
    {\bf \underline{Protocol~\customlabel{prot:#2}{\thefigure}: #1}}
    \end{center}
    
    #3
    \end{tcolorbox}
    }
}

\newcounter{datacounter}
\newcommand{\datastruct}[3]{
    \stepcounter{datacounter}
    \vspace{0.15cm}
    { \small
    \begin{tcolorbox}[breakable, enhanced, colback=darkpastelblue!20]
    \begin{center}
    {\bf \underline{Data Structure~\customlabel{def:#2}{\thedatacounter}: #1}}
    \end{center}
    
    #3
    \end{tcolorbox}
    }
}

\newcounter{coingame}

\newcommand\numberthis{\addtocounter{equation}{1}\tag{\theequation}}

\newcounter{casenum}
\newenvironment{caseof}{\setcounter{casenum}{0}}{\vskip.5\baselineskip}
\newcommand{\case}[2]{
    \refstepcounter{casenum}
    \ifthenelse{\equal{\value{casenum}}{0}}{
    \vskip.5\baselineskip\par\noindent
    }{}
    \noindent {\it Case \arabic{casenum}:} {\it #1}
    \vskip0.1\baselineskip
    \begin{addmargin}[1.5em]{1em}
    #2
    \end{addmargin}
}

\newcounter{subcasenum}
\newenvironment{subcaseof}{\setcounter{subcasenum}{0}}{\vskip.5\baselineskip}
\newcommand{\subcase}[2]{
    \refstepcounter{subcasenum}
    \vskip.5\baselineskip\par\noindent 
    {\it Subcase \arabic{casenum}.\arabic{subcasenum}:} {\it #1} \vskip0.1\baselineskip
    \begin{addmargin}[1.5em]{1em}
    #2
    \end{addmargin}
}
  
\newcounter{casenumb}
\newenvironment{caseofb}{\setcounter{casenumb}{0}}{\vskip.5\baselineskip}
\newcommand{\caseb}[2]{
    \refstepcounter{casenumb}
    \vskip.5\baselineskip\par\noindent 
    {\bf Case \arabic{casenumb}:} {\bf #1} \vskip0.1\baselineskip
    \begin{addmargin}[1.5em]{1em}
    #2
    \end{addmargin}
}

\newcounter{subcasenumb}

\begin{document}

\sloppy
\maketitle
\begin{abstract}

In the setting of error-correcting codes with feedback, Alice wishes to communicate a $k$-bit message $x$ to Bob by sending a sequence of bits over a channel while noiselessly receiving feedback from Bob. It has been long known (Berlekamp, 1964) 
that in this model, Bob can still correctly determine $x$ even if $\approx \frac13$ of Alice's bits are flipped adversarially. 
This improves upon the classical setting without feedback, where recovery is not possible for error fractions exceeding $\frac14$. 

\smallskip
The original feedback setting assumes that after transmitting each bit, Alice knows (via feedback) what bit Bob received.
In this work, our focus in on the \emph{limited feedback} model, where Bob is only allowed to send a few bits at a small number of pre-designated points in the protocol. For any desired  $\epsilon > 0$,
we construct a coding scheme that tolerates a fraction $ 1/3-\epsilon$ of bit flips relying only on $O_\epsilon(\log k)$ bits of feedback from Bob sent in a fixed $O_\epsilon(1)$ number of rounds. We complement this with a matching lower bound showing that $\Omega(\log k)$ bits of feedback are necessary to recover from an error fraction exceeding $1/4$ (the threshold without any feedback), and for schemes resilient to a fraction $1/3-\epsilon$ of bit flips, the number of rounds must grow as $\epsilon \to 0$.

\smallskip
We also study (and resolve) the question for the simpler model of erasures. We show that $O_\epsilon(\log k)$ bits of feedback spread over $O_\epsilon(1)$ rounds suffice to tolerate a fraction $(1-\epsilon)$ of erasures. Likewise, our $\Omega(\log k)$ lower bound applies for erasure fractions exceeding $1/2$, and an increasing number of rounds are required as the erasure fraction approaches $1$. 



\end{abstract}
\thispagestyle{empty}
\newpage

\enlargethispage{1cm}
\tableofcontents
\pagenumbering{roman}
\newpage
\pagenumbering{arabic}

\parskip=0.5ex
\section{Introduction}
Consider the problem of error-correcting codes with feedback~\cite{Berlekamp64}: Alice has a message $x\in \{0,1\}^k$ that she wants to communicate to Bob. She does this by sending a sequence of $n$ bits to Bob. After each bit she sends, Bob \emph{noiselessly} sends her some information about what he has received thus far. This \emph{feedback} can be used by Alice to adaptively choose the future bits she sends, so that she's always sending the most relevant bit of information at any point.
The goal is to do this in a way such that even if some fraction $\alpha$ of Alice's bits are corrupted, Bob can still correctly determine $x$. We remark that Bob's feedback bits are not corruptible and do not count towards the adversary's corruption budget.

Consider first the case of errors (the adversary may flip bits). What is the maximum fraction $\alpha$ of errors that an error-correcting code with feedback can be resilient to? In the standard setting without feedback, the answer is $\frac14$.\footnote{This result is usually stated for error-correcting codes where Eve can choose what bits to corrupt after receiving the entire message. However, in our setting, we require that Eve choose to corrupt each bit immediately after Alice sends it. Because we are in the deterministic setting, these models are the same.} Berlekamp~\cite{Berlekamp64} showed in his 1964 thesis ``Block Coding with Noiseless Feedback'' that in the feedback model, the answer increases to $\frac13$. In the case of erasures, standard error-correcting codes (without feedback) are resilient to $\frac12$ fraction of erasures, while it is easy to see (and explicitly noted e.g. in~\cite{GuptaKZ22}) that in the feedback setting one can tolerate an erasure fraction approaching $1$. 

In most prior work studying feedback, Bob is allowed to send as much feedback as he wants, and as often as he wants (in particular, after every bit sent by Alice).\footnote{In fact, feedback was previously modeled as Bob simply sending Alice the bit he received.} This results in $\Omega(k)$ bits of feedback sent over $\Omega(k)$ rounds. However, in practice, guaranteed error-free communication (i.e. noiseless feedback) is costly, as is interaction. In this paper, we ask the question: \emph{How much feedback does Bob really need to send to attain better error resilience than standard error-correcting codes?} More specifically, we ask the following two questions:\footnote{We mention that our protocols are all nonadaptive -- that is, Bob decides how much and when to send feedback before the start of the protocol. This prevents the parties from communicating information for free through silence.}
\begin{enumerate}
\itemsep=0ex
\item 
    How many \emph{bits} of feedback must Bob send?
\item 
    How many \emph{times throughout the protocol} must Bob send feedback?
\end{enumerate} 

\noindent\textbf{How many \emph{bits} of feedback?}
This first question was studied by~\cite{HaeuplerKV15} in the case of errors. In their work, they considered limiting Bob's feedback to $\delta$ fraction of Alice's bits. They showed that for randomized protocols, Bob sending any $\delta>0$ fraction \emph{as many bits as Alice} is sufficient to achieve the threshold $\frac13$ error resilience. For deterministic protocols, they showed that if $\delta > \frac23$ is a sufficiently large fraction, $\frac13$ is achievable, and as $\delta \rightarrow 0$, the best error resilience they achieve is $\frac14$.

In this work, we consider an even more limited form of feedback: we allow Bob to send only $d$ \emph{total} bits rather than $\delta$ fraction of Alice's bits, where $d$ does not scale proportionally to the number of bits Alice sends (which necessarily must be more than $k$ bits) or even to $k$.
In particular, we ask if we can achieve error resilience $\frac13$ with $o(k)$ bits of feedback.


Perhaps surprisingly, we show the answer is yes! Specifically, we construct a protocol that uses only $O(\log k)$ bits of feedback to achieve $\frac13$ error resilience. We also construct a protocol achieving $1$ erasure resilience with $O(\log k)$ bits of feedback. We complement our results with a lower bound showing that $\Omega(\log k)$ bits of feedback are necessary to even have $> \frac14$ error resilience or $> \frac12$ erasure resilience, thereby proving $\Theta(\log k)$ bits of feedback is optimal.

\medskip\noindent\textbf{How many \emph{rounds} of feedback?}
The second question we address concerns the number of \emph{rounds} of feedback necessary to achieve $\frac13$ error resilience. Previous works \cite{Berlekamp64,Muthukrishnan94,HaeuplerKV15,SpencerW92} all required $\Omega(k)$ rounds of feedback. Our result that only $O(\log{k})$ bits are necessary already implies $O(\log{k})$ rounds suffice. But can we limit this even further? After all, interaction is costly.

We show that a constant number of rounds of feedback suffices to obtain the threshold $\frac13$ error resilience, as well as to obtain $1$ erasure resilience. These constructions require only $O(\log{k})$ bits, thereby meeting optimality for the first and second question simultaneously.

\subsection{Our Results}

Our first result is a feedback ECC that is resilient to $\frac13$ errors (which is optimal) using only $O(\log k)$ bits of feedback over a constant number of rounds. This resolves an open problem posed by~\cite{HaeuplerKV15} asking whether deterministic feedback ECC's where the number of bits that Bob gives as feedback is a vanishingly small fraction of Alice's bits can achieve $\frac13$ error resilience. We use $O_\epsilon(\cdot)$ to hide constant factors depending on $\epsilon$ (these are bounded by a polynomial in $1/\epsilon$ in all our uses).

\begin{theorem}
    For any $\epsilon > 0$, there exists a deterministic feedback ECC that is resilient to errors in $\frac13 - \epsilon$ of Alice's communication, where Alice sends $O_\epsilon(k)$ bits, Bob's feedback consists of $O_\epsilon(\log k)$ bits sent in $O_\epsilon(1)$ rounds, and both parties have $\tilde{O}_\epsilon(k)$ computational complexity. 
\end{theorem}

In the case of erasures, we obtain similar parameters.

\begin{theorem}
    For any $\epsilon > 0$, there exists a deterministic feedback ECC that is resilient to $1 - \epsilon$ of Alice's communication being erased, where Alice sends $O_\epsilon(k)$ bits, Bob's feedback consists of $O_\epsilon(\log k)$ bits sent in $O_\epsilon(1)$ rounds, and both parties have $\tilde{O}_\epsilon(k)$ computational complexity. 
\end{theorem}

Next, we provide a matching lower bound on the total number of feedback bits necessary.

\begin{theorem}
    For any $\delta > 0$, any feedback ECC that achieves error resilience $\frac14 + \delta$ requires $\Omega_\delta(\log k)$ bits of feedback. In addition, any feedback ECC that achieves erasure resilience $\frac12 + \delta$ requires $\Omega_\delta(\log k)$ bits of feedback.
\end{theorem}

Finally, we show that in order to obtain error resilience arbitrarily close to $\frac13$, or erasure resilience arbitrarily close to $1$, the number of rounds of feedback necessary must grow arbitrarily large. Thus, it is necessary for the number of rounds to depend on $\epsilon$.

\begin{theorem}
    For any constant $r \in \bbN$, there exists $\delta(r) > 0$ such the error resilience of any $r$-round feedback ECC cannot exceed $\frac13 - \delta(r)$, and the erasure resilience of any $r$-round feedback ECC cannot exceed $1 - \delta(r)$.
\end{theorem}

Our results pin down (up to constant factors) the amount of feedback needed for \emph{deterministic} protocols. The optimal amount of feedback required in the randomized setting might have a completely different answer that is an interesting direction for future work.

\subsection{Related Works}

We already mentioned the seminal work of~\cite{Berlekamp64} which introduced the concept of error-correcting codes with feedback. Berlekamp studied adversarial corruption strategies and proved that $\frac13$ was the optimal error resilience for a feedback ECC in the face of adversarial errors (see Appendix~\ref{sec:appendix-1/3}). We also mentioned the work of~\cite{HaeuplerKV15} which initiated the study of limited feedback and the resulting error resilience.

Feedback ECC's have also been studied for the binary symmetric channel, where every bit is flipped independently with some fixed probability~\cite{Berlekamp68,Zigangirov76}. 


Many other works on error-correcting codes with feedback present themselves instead under the lens of \emph{adaptive search with lies}. The two setups are equivalent, and our results imply results in the field of adaptive search with lies. We discuss the relationship below.

\subsubsection{Adaptive Search with Lies} \label{sec:related-adaptive-search-lies}

An equivalent way of looking at error-correcting codes with unlimited feedback is as follows: 
Bob's feedback after every message from Alice consists of a question about Alice's input, and Alice's next bit is a response to that question. Then, the adversary can corrupt Alice's responses, but only up to some fraction $\alpha$ of the time. Bob is ``searching'' for Alice's input in a range $[N]$ by asking these yes/no questions, and receiving an incorrect answer up to an $\alpha$ fraction of the time. This game is known as the Rényi-Ulam game, independently introduced by Rényi in \cite{Renyi61} and Ulam \cite{Ulam91}. 

Two excellent surveys on the research in this area are presented in \cite{Deppe07,Pelc02}. There are a number of works that aspire to improve upon \cite{Berlekamp64}. For example, \cite{Muthukrishnan94} presents a strategy with optimal rate (asking the least possible number of questions) while still achieving the best error-correction threshold. An alternative construction meeting the same optimal error resilience was given in~\cite{SpencerW92}. We mention that some works focus on a fixed (constant) number of lies rather than resilience to a fraction of errors, and ask for the smallest number of questions necessary. \cite{Pelc87} shows the optimal bound for one lie, the original Rényi-Ulam game. Followups find the optimal strategy for two and three lies as well \cite{Guzicki90,Deppe00}. 

In our context, we consider the setting where Eve can corrupt $\alpha$ fraction of Alice's bits, and limit the number of rounds in which Bob  can send feedback and the number of feedback bits. In the context of adaptive searching with lies where Eve can corrupt $\alpha$ of Alice's answers, this translates to the following two limitations:
\begin{enumerate}
\itemsep=0ex
    \item What is the minimum number of \emph{bits} Bob must use to specify his questions?
    \item What is the minimum number of \emph{batches} in which Bob can ask his questions?
\end{enumerate}
The question of batches was considered in \cite{NegroPR95,AhlswedeCDV09}. The work of \cite{Pelc02} poses the open question of how many batches of questions are necessary before the optimal strategy becomes essentially the same as the fully adaptive setting. Though ``essentially the same'' is not well-defined here, our result shows that Bob needs only a constant number of batches of questions to achieve the maximum possible error resilience, with only a constant factor hit to the total number of questions.

Many works consider the restricted version of the Rényi-Ulam game, where Bob is only allowed to ask a limited set of questions to query Alice's input in $[N]$. For example, \cite{MengLY13, Macula97} considers the case where Bob must ask if Alice's input is in a set of limited size. The works \cite{RivestMKWS80, NegroPR95} and others consider comparison questions (``is $x<a$?''). Note however that both these papers target a goal slightly different from ours, namely to minimize the number of questions for a fixed number of lies, not to achieve an error resilience of $\frac13$. Though more restrictive than allowing Bob to ask any question of his choice, the common theme is that these questions still require $O(k)$ bits for Bob to specify. Our result shows a strategy where Bob is limited to questions he can specify with only $O(\log k)$ bits.

\subsubsection{Interactive Error-Correcting Codes}

In the setting of error-correcting codes with feedback, one measures noise resilience as a fraction of just Alice's communication. That is, Bob's feedback is not corruptible and is not counted towards the adversary's corruption budget. 

In the model of \emph{interactive error-correcting codes (iECCs)}, Alice's goal is still to communicate a message $x$ to Bob, and Bob is allowed to send feedback, but the corruption budget is measured as a fraction of the \emph{total communication}, including both Alice's forward messages and Bob's feedback. The adversary can choose how to distribute their corruption budget among the forward and feedback communication: they can corrupt only Alice's bits, or only Bob's, or some combination of the two. Since the adversary can corrupt either party, iECCs are essentially codes with \emph{noisy} feedback rather than noiseless feedback. 

The question of whether iECCs could perform better in terms of noise resilience than simply sending a standard error-correcting code was posed by~\cite{HaeuplerKV15}. This question was answered in the affirmative by~\cite{GuptaKZ22} in the case of erasures: they demonstrated an iECC that achieved $\frac35$ erasure resilience (whereas a standard ECC can't correct a fraction $1/2$ of erasures). The followup work of~\cite{GuptaZ22b} constructed a related scheme for erasures with positive rate, but only resilient to a fraction $\frac6{11}$ of erasures.
In the case of errors, the work of~\cite{EfremenkoKSZ22} constructed an iECC beating the bound of $\frac14$ given by a standard error-correcting code.

We remark that for both erasures and errors, the optimal noise resilience of an iECC remains an open problem. The best known upper (impossibility) bounds are $\frac23$ in the case of erasures, and $\frac27$ in the case of errors~\cite{GuptaKZ22}.

Interactive error-correcting codes are also closely related to the broader topic of coding for interactive communication~\cite{Schulman92,Schulman93,Schulman96}. This field also often draws ideas from error-correcting codes with feedback. An excellent survey of the field is presented in \cite{Gelles-survey}.

\section{Technical Overview}
We now give an overview of our results.

Suppose Alice has an input $x \in \{ 0, 1 \}^k$. She communicates $x$ to Bob over a series of rounds, after each of which she receives some (uncorrupted) feedback from Bob that she can use in choosing her future messages to Bob. The goal is to do this in a way so that even if an adversary corrupts a large fraction of Alice's communication to Bob, he can still determine $x$ unambiguously. 

The information-theoretically largest fraction of Alice's communication that may be corrupted under which Bob can still determine $x$ correctly depends on whether the type of corruption is errors (bit flips) or erasures. In the case of errors, the threshold is $\frac13$, and in the case of erasures, the threshold is $1$. We begin with the simpler construction for erasures in Section~\ref{sec:overview-erasure} and then describe the more challenging construction for errors in Section~\ref{sec:overview-error}. Finally, we sketch the lower bounds on the number of bits and number of rounds in Section~\ref{sec:overview-lb}.
\subsection{Feedback ECC for Erasures} \label{sec:overview-erasure}

We describe our feedback ECC resilient to $1-\epsilon$ erasures.

Let $\ECC$ be an error-correcting code encoding messages of length $k$ with list-decoding radius $1-\epsilon$. The protocol consist of $O_\epsilon(1)$ rounds of length $|\ECC|$. 

Alice begins by sending $\ECC(x)$ to Bob every round until Bob receives a message with at most $1 - \epsilon$ erasures. By the list-decoding guarantee, Bob can narrow $x$ down to $L_\epsilon = O_\epsilon(1)$ options. At this point, Bob can simply send Alice the list of $L_\epsilon$ options via noiseless feedback, so that Alice simply has to send a number in $[L_\epsilon]$ to tell Bob which of the $L_\epsilon$ options is her input over the remaining rounds. However, sending over the full list of $L_\epsilon$ options takes $L_\epsilon \cdot k$ bits, which is too much. Instead, Bob sends over a list of $L_\epsilon$ indices on which his $L_\epsilon$ options differ, as well as the evaluation of each possible message on the indices in sequence. This allows Alice to uniquely determine which of the $L_\epsilon$ options is hers, so she knows which number from $[L_\epsilon]$ to communicate back to Bob. 

The second step is for Alice to communicate this number from $[L_\epsilon]$ to Bob using the remaining rounds of the protocol. In particular, Alice's task has been reduced from communicating a number in $[2^k]$ to a constant-sized number. This can be done via a $(1-\epsilon)$-erasure resilient feedback ECC that takes $O_\epsilon(L_\epsilon)$ rounds, noted e.g. in~\cite{GuptaKZ22}. In total, this gives a constant round feedback ECC resilient to a fraction $1-\epsilon$ of erasures.

This strategy crucially relies on the fact that when Bob receives a message with at most $1-\epsilon$ erasures that he decodes to $L_\epsilon$ options, he knows \emph{for sure} that Alice's message is one of the options. In the case of errors, Bob does not have this guarantee. Even if he decodes a message to $L_\epsilon$ options, it is possible all the options are incorrect, if the adversary happened to corrupt more than $\frac12-\epsilon$ fraction of the bits in that round. For this reason, a similar protocol does not work in the case of errors.

\subsection{Feedback ECC for Errors} \label{sec:overview-error}

We begin by describing past work that is useful context for our protocol. Most past works on feedback ECC's or adaptive search with lies (see Section~\ref{sec:related-adaptive-search-lies}) in the case of errors rely on a certain reformulation of the ECC with feedback problem, which Spencer and Winkler~\cite{SpencerW92} call the \emph{coin game}. We mention that there is another approach in the literature for achieving $\frac13$ error resilience~\cite{HaeuplerKV15}, known as the \emph{rewind-if-error} paradigm (see Appendix~\ref{sec:appendix-rewind-if}). However, it seems that the rewind-if-error paradigm is not conducive to bit minimization. 


\medskip\noindent\textbf{The coin game.}
In the coin game setup, there are $2^k$ coins corresponding to the $2^k$ possible values of Alice's input on a board. At the beginning of the game, all coins are at position $0$. Throughout the game, Alice sends Bob a message, after which Bob updates the positions of all coins on the board with the amount of corruption Eve would have needed to use in that message if Alice actually had that coin as her input. If the protocol has $|\pi|$ bits, we say that a coin falls off the board if its position exceeds $|\pi|/3$ (since then it is no longer a valid possibility for Alice's input if the adversary is constrained to corrupt $\lesssim$\footnote{We use $\lesssim$ and $\gtrsim$ to mean that the inequality holds within a factor of $(1-O(\epsilon))$} $|\pi|/3$.). The goal is to design Alice's messages in response to Bob's feedback in such a way that only one coin can remain on the board after $|\pi|$ rounds, regardless of what Eve corrupts Alice's messages to.

\medskip\noindent\textbf{The Spencer-Winkler feedback ECC~\cite{SpencerW92}.}
To put our construction into context, we first discuss the construction of Spencer and Winkler. At a high level, at every step, Bob tells Alice an ordering of his coins from smallest to largest position value (breaking ties arbitrarily). In response, Alice sends a single bit: $1$ if her input is an odd numbered coin according to this ranking, and $0$ if her input has an even ranking. Then, whenever Bob receives a $1$, he moves the position of all the even coins up by $1$, and if he receives a $0$, he moves the position of all the odd coins up by $1$. Note that the set of coins with an even/odd ranking change after each update.

At a high level, Spencer-Winkler's analysis is as follows. Let $\posx(i)$ denote the position of the $i$'th coin in this ranking. It turns out that regardless of whether Eve chooses to move the even or odd coins each up by $1$, the quantity $\posx(1)+\posx(2)+\posx(3)$ increases by $\gtrsim 1$ at each step. This shows directly that all but two of the coins must fall off the board (position exceeds $|\pi|/3$) by the end of the protocol. To show that only one coin can remain, we note that as soon as the third to last coin falls off the board, as long as the two remaining coins stay on the board, their combined position increases at $1$ each round. Then, if we consider the potential function $\posx(1) + \posx(2) + \min \{ \posx(3), |\pi|/3 \}$, we see that this potential must increase by $\gtrsim 1$ each round, so that it is $\gtrsim |\pi|$ by the end of the protocol. This means that by the end of the protocol, $\posx(2) \gtrsim |\pi|/3$, ultimately ensuring that the second coin falls off the board as well.

We make a few remarks about the Spencer-Winkler protocol.

\begin{itemize}
\itemsep=0ex
\item 
    In every round, Bob increments the positions of exactly half the coins. In particular, half the coins previously at position $0$ will remain at position $0$ after a round. This means that in order to distinguish between $2^k$ coins, at least $k$ rounds are necessary. 
    

\item 
    It is important that the first two coins are separated by distance $1$ in Alice's message. This plays an important role after the third to last coin has fallen off the board in ensuring that the second to last coin also falls off. However, before the third to last coin falls off, the important metric is $\posx(1) + \posx(2) + \posx(3)$ increasing by $1$. 
    In our feedback ECC, we use a different strategy to achieve this, that does not involve separating the first two coins by distance $1$. 
    
    

\end{itemize}






\subsubsection{Our feedback ECC}
At a high level, Alice and Bob alternate among three strategies throughout the protocol. In each section of the protocol, which we call a chunk (of which there are $C \approx \frac1\epsilon$ many), Bob tells Alice which strategy the duo will be using during that chunk based on the current positions of his coins. The chunks are further subdivided into a constant number $R$ of rounds, which Alice and Bob use to execute the strategy. In each round, Bob sends Alice $O(\log{k})$ bits of feedback, and then Alice sends a message of length $M = O(k)$. The three strategies are as follows:

\begin{enumerate}
\itemsep=0ex
    \item\label{strat:ecc} {\bf Alice sends a traditional $\ECC(x)$ for the entire chunk:} 
        In this strategy, Alice ignores Bob's feedback between rounds and sends $\ECC(x)$ of length $M$, where $\ECC$ is an error-correcting code with relative distance $\frac12$. This strategy has the property that the positions of at most a small (constant) number of coins may increment by $\lesssim \frac13 M$ each round, and by list-decoding guarantees, most increment by $\gtrsim \frac12 M$ (recall that our goal is to show that all except one coin reaches position $|\pi|/3$ by the end of the protocol, so this means only a few coins are not on track). 
    \item\label{strat:mini} {\bf Alice and Bob play a special coin minigame:} 
        In this strategy, Alice and Bob play a coin game somewhat similar to Spencer and Winkler's even/odd strategy for a constant number of rounds (the duration of the chunk). This minigame also increments most coins by $\gtrsim \frac13 M$, but guarantees to separate the smaller coins by a larger amount than the traditional $\ECC$ strategy.
        
        The exact design of this game is more complicated than the even/odd strategy, and the fact that it compensates exactly for the shortcomings of the $\ECC(x)$ strategy is perhaps the most interesting part of our protocol.
    \item\label{strat:bit} {\bf Alice sends a single bit separating Bob's first and second coin.} 
        Once all but two coins fall off the board, for the rest of the protocol, Alice sends a single bit to achieve relative distance $1$ between the two remaining coins. This part of the protocol is very similar to Spencer and Winkler's original construction.
\end{enumerate}

Alice and Bob spend most of the protocol alternating between strategies \ref{strat:ecc} and \ref{strat:mini}, before eventually using strategy \ref{strat:bit} for the rest of the protocol. Bob chooses between the first two strategies as follows: if $\posx(3)\approx \posx(L)$, where $L$ is some large constant that we choose later,\footnote{$L$ technically scales with the round (chunk) number, but because there are a constant number of rounds, it is not an important subtlety. $L$ intends to approximate an ``average but relevant'' coin.} he chooses strategy \ref{strat:ecc}, and otherwise he chooses strategy \ref{strat:mini}. 

In the first two strategies, Alice and Bob's goal is to have $\posx(1)+\posx(2)+\posx(3)$ increase at a rate of $\gtrsim RM$ per chunk (recall that each chunk has $R$ rounds of length $M$). If they can accomplish this, then we can use the argument of Spencer and Winkler described above: eventually the third coin falls off the board, and at that point, since strategy~\ref{strat:bit} increases $\posx(1)+\posx(2)$ at a rate of $\gtrsim RM$ per chunk, the second coin will fall off the board as well.

We discuss briefly how the third strategy is implemented. For feedback, Bob sends to Alice an index $i \in [k]$ for which the first two coins differ. Then, Alice simply sends $x[i]$ for the entire duration of the chunk. As such, the distance between what Alice would've sent if her input were each of the two coins is $RM$ (per chunk), and so regardless of what Eve corrupts Alice's messages to, $\posx(1)+\posx(2)$ increases at a rate of $RM$ per chunk. Then, by the same argument used in Spencer and Winkler's protocol, the second coin falls off the board as well.


For the rest of the section, the goal is to ensure that the first two strategies increase $\posx(1)+\posx(2)+\posx(3)$ at a rate of $\gtrsim RM$ per chunk.

\subsubsection{Strategy~\ref{strat:ecc}: The traditional distance $\frac12$ ECC} 

Alice spends the entire chunk sending $\ECC(x)$ to Bob, which is a message of length $RM$. Bob's feedback between rounds is just ignored (the round structure is not necessary for this strategy, but is included in the overall protocol structure for uniformity because it is needed in the strategy~\ref{strat:mini}). At the end of the chunk Bob updates the coin states, moving each coin $y$ by $\Delta(m,\ECC(y))$ where $m$ is the corrupted length $RM$ message he received throughout the chunk.\footnote{As described, this update takes exponential-in-$k$ time, but there is a data structure for this that we describe in Section~\ref{sec:coin-game-def}.}

By the list decoding properties of traditional ECC's, one can guarantee that all but a constant number of coins increase by almost $\frac12RM$ in position. However, there may be a small number of coins that do not increase by at least $\frac12RM$ in value. For example, the adversary may choose a corruption pattern so three of the coins' positions increase by only $\frac14RM$ in every chunk.  This causes $\posx(1)+\posx(2)+\posx(3)$ to increase only at a rate of $\gtrsim \frac34 RM$ per chunk -- which is not enough. 

What is nice though is that $\posx(1)+\posx(2)+\posx(L)$ increases at a rate of $\gtrsim RM$ per chunk. As such, if $\posx(1)+\posx(2)+\posx(3)$ isn't increasing fast enough, then $\posx(3)$ and $\posx(L)$ are separating out and we'll quickly exit this strategy.

\subsubsection{Strategy~\ref{strat:mini}: The coin minigame}

Our main goal here is to separate out Bob's first $L$ coins further, while the third coin and $L$'th coin are separated. To do this, we will play a new coin game essentially using just the first $L$ coins.\footnote{Technically, the new game will be played on a \emph{partition} of $\{ 0, 1 \}^k$ separating out the first $L$ coins, as we will describe shortly.} The point is that since the third and the $L$'th coins are separated whenever this strategy is in play, the end behavior of the first three coins in the original game is determined entirely by the new coin game on the $L$ coins, i.e. no later coins can become one of the first few and interfere with them.

Our strategy requires Bob to single out the first $L$ coins. However, this requires $O(k)$ bits of communication, and we only have $O(\log{k})$ bits -- so Bob can't directly tell Alice what these $L$ elements are. Instead, he tells her a partition $\cP$ of the $2^k$ coins into $L$ sets $P_1, \dots, P_{L}$ so that coins $1\ldots L$ (as determined by smallest positions) are in different sets. He does this with only $O(\log{k})$ bits by giving her a constant-sized list of indices on which coins $1\ldots L$ all differ, along with the evaluations of the first $L$ coins on these indices. 

During the chunk, Alice and Bob play a new instantiation of a coin game, called the \emph{coin minigame}, on the \emph{sets} of $\cP$ (so coins within the same set behave as a single unit). This coin minigame has its coins as the $L$ sets $P_1, \dots, P_{L}$, all starting at position $0$. Throughout the chunk, Bob moves the position of each set in the coin minigame according to Alice's message. The position of a set in the coin minigame is the amount of corruption caused to coins in that set. Bob will update the original coin game with the amount of corruption seen in his received messages if Alice actually had a given coin, as before. The hope is to have $\posx(1) + \posx(2) + \posx(3)$ in the original coin game increase at rate $\gtrsim RM$.

\medskip\noindent\textbf{Attempt using Spencer and Winkler's even/odd strategy.} We describe an attempted version of this coin minigame using the even/odd strategy. Before each round, Bob orders the sets $P_1, \dots, P_{L}$ according to their positions in the coin minigame as $S(1)\ldots S(L)$, where $S(i)$ is the set in the $i$'th smallest position of the minigame, and tells Alice this ordering. In each round, Alice tells Bob whether the set containing her coin has an even or odd ranking according to this ordering, and Bob moves his coins accordingly. Then, at the end of the chunk, he adds the position of each set to each coin in that set. 

Then, by the end of the minigame chunk, it would hold that the three smallest increments sum to $\gtrsim RM$. This is progress -- in the main game, this means the quantity $\posx(1)+\posx(2)+\posx(3)$ increases by $\gtrsim RM$ during the chunk. This crucially requires that the first three coins now must have all been in the first $L$ coins at the start of the chunk -- which holds because $\posx(L)$ was much larger than $\posx(3)$ at the start of the chunk. Therefore, the first three coins were in different sets in the coin minigame, so the total increase was at least $RM$. 

Unfortunately, this does not quite finish our protocol design, since strategy~\ref{strat:ecc} wasn't actually causing the correct rate of increase of $\posx(1)+\posx(2)+\posx(3)$. Concretely, during strategy~\ref{strat:ecc}, $\posx(1)+\posx(2)+\posx(3)$ increases at a rate of $\approx \frac34 RM$, and $\posx(1)+\posx(2)+\posx(L)$ increases at a rate of $\approx RM$ per chunk. In strategy~\ref{strat:mini}, $\posx(1)+\posx(2)+\posx(3)$ increases at a rate of $\approx RM$, but $\posx(1)+\posx(2)+\posx(L)$ might increase slower (as it turns out, it can increase as slowly as $\approx \frac12$). Over time, these strategies will average out, and $\posx(1)+\posx(2)+\posx(3)$ increases at some rate strictly between $\frac34 RM$ and $RM$, which is not fast enough. 

Viewing $\posx(L)$ as a proxy for $\posx(3)$ (recall that we only enter strategy~\ref{strat:ecc} when the third and $L$'th coin are close), what we need is for strategy~\ref{strat:mini} to increase $\posx(1)+\posx(2)+\posx(L)$ at a rate of $\approx RM$ per chunk as well. Then, one can argue that up until the last time Alice and Bob use strategy~\ref{strat:ecc}, $\posx(1) + \posx(2) + \posx(L)$ and thus $\posx(1) + \posx(2) + \posx(3)$ has increased at rate $\gtrsim RM$ each chunk, thus correcting the issue with strategy~\ref{strat:ecc} that $\posx(1) + \posx(2) + \posx(3)$ doesn't increase fast enough.

\medskip\noindent\textbf{A different coin minigame.} To change the strategy to one where $\posx(1) + \posx(2) + \posx(L)$ at $\gtrsim RM$ as well, we need a different coin minigame strategy than the even/odd one. Recall that Alice and Bob play this coin minigame on the $L$ sets of a partition $P$. In every round, Bob tells Alice his current ordering of the sets, namely $S(1)\ldots S(L)$, where $S(i)$ is the set in the $i$'th smallest position.

In our new coin minigame, instead of telling Bob even/odd, Alice sends Bob one of three things: $m_1$ if $x\in H_1 := S(1)$, $m_2$ if $x\in H_2 := S(2) \cup S(4) \cup S(6) \cup \ldots$, or $m_3$ if $x\in H_3 := S(3) \cup S(5) \cup S(7) \cup \ldots$. The three options $m_1,m_2,m_3$ can be encoded into distance $\gtrsim \frac23 M$ apart messages (for instance, $011, 101, 110$ repeated $M/3$ times). As it turns out, this new strategy also results in the smallest three sets increasing by a total of at least $RM$ per chunk, so because the $L$'th coin is far enough away from the third one at the beginning of the chunk, the main coin game looks like the minigame for the smallest coins, so $\posx(1)+\posx(2)+\posx(3)$ increases by $\gtrsim RM$. In fact, this new strategy has the stronger property that the smallest two sets increase by $\frac23 M$ per chunk. However, this only implies that $\posx(1)+\posx(2)+\posx(L)$ increases by $\gtrsim \frac23 RM$, which is not enough.

The main problem is still that $\posx(L)$ may be increasing at a rate of $0$, while we need it to increase at $\gtrsim \frac13 RM$. This is inevitable with the ideas we've discussed so far. Although we can guarantee that the behavior of the first three coins is correct since the $L...2^k$'th coins are far enough away at the beginning of the chunk, the issue is many of the $L$'th through $2^k$'th coins are in the same set of the partition $P$ and could have all increased by $0$. 

The solution is to separate out elements within the same set $H_i$ as well. Instead of sending $m_i$, where $i\in \{1,2,3\}$ corresponds to which set $H_i$ $x$ is in, Alice sends the encoding $\ECC(i,x)$, where elements with different values of $i$ are $\frac23 M$ apart, and elements with the same value of $i$ have some nonzero distance $\delta_i M$ as well. Bob increments the positions of the sets in the coin minigame by the minimum distance between the message he received and any coin in that set, but increments the positions of each coin in the main game individually according to how much corruption there must've been. Then, coins within the same set $H_i$ will not all increment uniformly: depending on the individual distance from the received message, coins in the main game will increment by different amounts.

If we can take $\delta_i = \frac13$ for $i = 1, 2, 3$, then we would be done. To see this, we can appeal once again to list-decoding, so that all but a constant number of coins increase by $\gtrsim \frac13 M$ in each round (and thus $\gtrsim \frac13 RM$ by the end of the chunk). Then, combining with the fact that $\posx(1) + \posx(2)$ has increased by $\gtrsim \frac23 RM$, it follows that $\posx(1) + \posx(2) + \posx(L)$ increases by $\gtrsim RM$ each chunk.\footnote{We remark again that $L$ is technically scaling with the chunk number, but we omit those details here. The important point is that $L$ represents a sufficiently large rank in the coin game that describes the behavior of most coins in the game.}  


However, it turns out that achieving distance $\frac13 M$ within each of the three sets is impossible. A key observation is it is only necessary to have nonzero distance for $i = 1$, because all coins in $H_2$ or $H_3$ must necessarily have sufficiently high positions already, so they don't need the extra distance. More precisely, we define the code $\ECC$ used in this round as follows. Let $\codeC(y)$, $y\in S(1)$ be a code with relative distance $\frac13$. Let $m_2$ be a message that is relative distance $\frac23$ from all elements of $\codeC$. Let $m_3$ be relative distance $\frac23$ from $\codeC$ and from $m_2$. Then
\[
\ECC(y) = \begin{cases}
    \codeC(y) & y\in S(1) \\
    m_2 & y\in S(2),S(4),S(6),\ldots \\
    m_3 & y\in S(3),S(5),S(7),\ldots \\
\end{cases}\]
Upon receiving a message $m$, Bob updates each coin $P_i$ in the minigame based on $\min_{y\in P_i}\Delta(m,\codeC(y))$, and each coin $y$ in the main game by $\Delta(m,\codeC(y))$.

Using this encoding, it holds that all but a constant number of coins increase by $\gtrsim \frac13 RM$ in every chunk. Then, we get that $\posx(1)+\posx(2)+\posx(L)$ increases by $\gtrsim RM$, and we are done!


\subsubsection{Analysis of Coin Minigame} \label{sec:overview-minigame}

To wrap up the discussion of our feedback ECC resilient to $\frac13$ errors, we briefly discuss the analysis of the coin minigame, in particular our earlier claim that the two smallest positions among the sets $P_1, \dots, P_{L}$ increase by $\gtrsim \frac23$ each round. This innocuous claim is surprisingly difficult to show, and a nontrivial part of our contribution. For a detailed analysis, we refer the reader to Section~\ref{sec:minigame-analysis}.


The analysis of the even/odd protocol in \cite{SpencerW92} goes as follows. They first prove that it is always in the adversary's best interest to corrupt Alice's message to $0$, indicating that $x$ is in the ``evens'' set, regardless of what she sends, until the third coin falls off the board. Then, they analyze the positions of the coins in the particular case that Bob always receives the bit $0$, and show that the quantity $\posx(1)+\posx(2)+\posx(3)$ increases at a rate of $\gtrsim 1$ as required. 

Unfortunately, in our protocol, we are not able to show a strictly dominant corruption pattern for the adversary, and so we must analyze the adversary's various corruption patterns simultaneously. We thus have to invent a new technique for analyzing a more continuous spectrum of attacks, independent of previous work.

To analyze the minigame, we introduce the following new framework, which involves considering \emph{pairs} of coins as a single object.\footnote{We remark that this framework of considering a few coins together as a single object is very powerful: in fact, the feedback ECC of~\cite{SpencerW92} can alternatively be analyzed via a similar analysis considering \emph{triples} of coins instead of pairs.} In particular, we introduce a new derivative coin game in which the new coins are pairs of coins in the old game, and the position of a new coin is the sum of positions of the two coins in the old game. The goal is then to show that after a large constant $R$ number of rounds, all the new coins are at positions $\gtrsim \frac23 RM$.

While not all coins in this new pair game are increasing at $\frac23 M$ each round, on average they are increasing at $\frac23 M$ each round. This in itself is not enough, as one could imagine the coin with the smallest position updating at $< \frac23 M$ each round while a larger one updates at $>\frac23 M$ per round. However, we are able to show that if coin $(y, z)$ updates at $< \frac23 M$ in a round, there is a coin $(y', z')$ with \emph{smaller initial position} that updates with $> \frac23$ compensating for it. In particular, for any coin that lags, there is a coin with a smaller initial position that is overcompensating and catching up. This limits the largest possible lag to be small compared to the total number of rounds. We formally argue this in Section~\ref{sec:minigame-analysis}.
\subsection{Lower Bounds}\label{sec:overview-lb}

Finally, we sketch our lower bounds on the number of bits of feedback and the number of rounds of feedback.

\subsubsection{Number of Bits of Feedback}

We show that no feedback ECC can achieve an error resilience of $\delta>\frac14$ or erasure resilience of $\delta>\frac12$ without Alice sending $\Omega(\log{k})$ bits of feedback. We outline the argument for errors, and the argument for erasures is similar.

Assume for the sake of contradiction some feedback ECC achieves $\delta>\frac14$ error resilience. For every pair of inputs $x,y$ that Alice could have, we employ the following attack. Eve corrupts the string to look like what Alice would send if she had $x$, until this the number of corruptions in the case Alice has $y$ reaches $\delta |\pi|$, and then she corrupts the remainder to look like she had $y$. By assumption, this attack cannot work, and so must require more than $\delta$ corruption in the case Alice has $x$. Let the feedback string Bob sends in this be $s(x,y)$. It must hold that the messages Alice sends in the case she has $x$ or $y$ given the feedback $s(x,y)$ are separated by at least $2\delta$.

Now, by appealing to upper bounds for the multicolor Ramsey problem, one can show that for large enough $k$, there is a size $\left\lfloor \frac{10}{\delta-1/4}\right\rfloor$ set $S$ of Alice's possible inputs such that all pairs of inputs $x,y\in S$ have the same value of $s=s(x,y)$. In particular, this means that on the feedback string $s$, what Alice sends on every $x\in S$ is separated by at least $2\delta$. This contradicts the Plotkin bound~\cite{Plotkin60}.

We refer the reader to Section~\ref{sec:lb-bits} for a more detailed argument.

\subsubsection{Number of Rounds of Feedback}

We show that for any feedback ECC with a fixed number of rounds $r$ of feedback, the error resilience cannot approach $\frac13$ for errors or $1$ for erasures. More precisely, there exist $\delta_{\text{error}}(r)<\frac13$ and $\delta_{\text{erasure}}(r)<1$ such that the resilience does not exceed the corresponding value of $\delta(r)$. We outline the argument for errors, and the argument for erasures is similar.

Our argument proceeds by induction. For $r = 0$, the protocol is simply a non-interactive ECC, so $\delta_{\text{error}}(0) = \frac14 < \frac13$. Now supposing that an $r-1$ round feedback ECC has error resilience separated from $\frac13$, the goal is to show that an $r$ round feedback ECC has the same property. We employ a different attack based on whether the first time Bob sends feedback is early in the protocol or late in the protocol. If the first round of feedback is early, Eve corrupts Alice's first round of messages entirely (since the first round of messages is short, this is not very much). Then, Bob's first round of feedback is useless, since Alice has not sent any information yet, and the remainder of the protocol is an $r-1$ round feedback protocol, and so the adversary needs to corrupt $\delta(r-1)$ fraction of the remainder of the protocol. If the first round of feedback is later in the protocol, the adversary corrupts $\frac14$ of Alice's first round, which lets Bob narrow down Alice's input to three options. Then there is an attack requiring $\frac13$ corruption of the rest. Since the first round was long, overall this is still noticeably less than $\frac13$ corruption.

For the precise calculations of corruption incurred by these attacks, we refer the reader to Section~\ref{sec:lb-rounds}.

\section{Preliminaries and Definitions}\label{sec:prelims}

\paragraph{Notation.}

In this paper, we use the following notation:
\begin{itemize}
\itemsep=0ex
    \item The function $\Delta(x, y)$ represents the Hamming distance between $x$ and $y$.
    \item $x[i]$ denotes the $i$'th bit of a string $x \in \{ 0, 1 \}^*$. 
    \item $x[i:j]$ denotes the substring of $x$ consisting of the $i$'th to $j$'th bits.
    \item $x||y$ denotes the string $x$ concatenated with the string $y$.
    \item $[N]$ denotes the integers $1, \dots, N$. The notation $[i:j]$ denotes the integers $i, \dots, j$. 
    \item The complexity $O_\epsilon(n)$ means $f(\epsilon)\cdot O(n)$ for some function $f:\bbR \to  \bbN$, and $\tilde{O}_\epsilon(n)$ means $f(\epsilon)\cdot n (\log n)^c$ for some function $f:\bbR \to  \bbN$ and absolute constant $c < \infty$.\footnote{In our paper, one can restrict $f$ to be polynomial in $1/\epsilon$.}
\end{itemize}

\subsection{Error-Correcting Codes with Feedback}

\begin{definition} [Feedback ECC]
    A \emph{feedback ECC} $\code$ is a family of protocols $\{\pi_k\}_{k\in \bbN}$ between Alice and Bob with the following properties.
    \begin{itemize}
        \item Alice begins with an input $x\in \{0,1\}^k$ that she intends to communicate to Bob.
        \item The protocol consists of a number of rounds, in which Bob speaks followed by Alice. The number of bits communicated by each party in a given round is fixed beforehand. Bob's communication is known as \emph{feedback}. 
        \item At the end of the protocol, Bob outputs a guess for $x$.
        \item The number of (message) bits Alice sends is denoted $|\code_k|$ and the number of (feedback) bits Bob sends is denoted $\zeta(\code_k)$. The number of rounds in which Bob speaks is denoted $\rho(\code_k)$
    \end{itemize}
    The family of protocols is said to have error (resp. erasure) resilience $\alpha$ if, for sufficiently large $k$, for all online adversarial attacks where the adversary can flip (resp. erase) up to $\alpha$ fraction of the bits Alice communicates (and none of Bob's bits) in an online fashion, choosing whether to corrupt a message immediately after it is sent, Bob outputs $x$ correctly.
\end{definition}

We note that because our model is deterministic, Eve corrupting Alice's message immediately after each bit or after each round is the same -- Alice's next message is predictable by her input and the publicly available feedback Alice has received so far. in  Though Eve doesn't technically know Alice's input, any attack only needs to succeed on a single input, so any attack may assume Eve knows Alice's input.

    

\subsection{Error-Correcting Codes}

We now state, and briefly justify, the properties of list-decodable codes employed in our scheme. Below $\tilde{O}_\epsilon(n)$ hides constant factors depending on $\epsilon$ as well as $\poly(\log n)$ factors.

\begin{theorem}\label{thm:ecc} 
  Let $\epsilon > 0$. For all integers $k$, there exists $N_0 = N_0(k) \le O(k/\poly(\epsilon))$ such that for each $n_0 \in \{N_0,2N_0,3N_0\}$, there exists an explicit error-correcting code 
    \[
        \ecc_{\epsilon, k} : \{0,1\}^k \rightarrow \{ 0, 1 \}^{n_0} 
    \]
    with the following properties:
    \begin{itemize}
    \item 
        {\bf Distance close to $\frac12$:} 
        For all $x,y \in \{0,1\}^k$, $x \neq y$,
        \[
            \Delta \big( \ecc_{\epsilon,k}(x), \ecc_{\epsilon,k}(y) \big) \ge \left( \frac{1-\epsilon}{2} \right) \cdot n_0, \numberthis \label{eqn:ecc-1/2}
        \]
    \item 
        {\bf Efficient Encoding:} 
        The code is efficiently encodable. Specifically, for any $x \in \{ 0, 1 \}^k$, one can compute $\ecc_{\epsilon,k}(x)$ in $\tilde{O}_\epsilon(n_0)$ time. 
    \item 
        {\bf List Decoding for Errors:} 
       For any string $m\in \{0,1\}^{n_0}$, the number of values of $x\in \{0,1\}^k$ satisfying 
        \[
            \Delta \big( \ecc_{\epsilon,k}(x), m \big) \le \left( \frac{1-\epsilon}2   \right) \cdot n_0, \numberthis \label{eqn:ecc-list-1/2}
        \]
        is at most a constant $L_\epsilon = O_\epsilon(1)$ independent of $k$ and $n_0$. Moreover, there is an algorithm to compute these values in $\tilde{O}_\epsilon(n_0)$ time. Computing these values is called ``list decoding.''
    \item 
        {\bf List Decoding for Erasures:} 
        For all $x \in \{ 0, 1 \}^k$, if at most $1 - \epsilon$ fraction of $\ecc_{\epsilon,k}(x)$ is erased, there is a $\tilde{O}_\epsilon(n_0)$ time algorithm that outputs $L_\epsilon = O_\epsilon(1)$ values $x_1, \dots, x_{L_\epsilon} \in \{ 0, 1 \}^k$ with the guarantee that $x \in \{ x_1, \dots, x_{L_\epsilon} \}$. 
    \item 
        {\bf Assumption for Claim~\ref{claim:ecc3}:} We also make the assumption for ease of notation later that all codewords are also far from both $0^{n_0}$ and $1^{n_0}$; that is, for every $x \in \{0,1\}^k$
        \begin{align*}
            \Delta(\ecc_{\epsilon,k}(x), 0^{n_0}) \ge \left( \frac{1-\epsilon}2 \right) \cdot n_0  \qquad \text{and} \quad
            \Delta(\ecc_{\epsilon,k}(x), 1^{n_0}) \ge \left( \frac{1-\epsilon}2  \right) \cdot n_0.
        \end{align*}
    \end{itemize}
\end{theorem}
\begin{proof}
We give a brief sketch of the construction. First, we note that if we have such a code of length $N_0$, then we can simply repeat the encoding twice  to get a code of length $2 N_0$. We can list decode the first and second blocks of the received word and take the union, at most doubling the list size. To get a code of length $3 N_0$ we simply repeat the encoding three times.

If we have the list decoding guarantee for errors, then we can also get the guarantee for erasures by filling in the erasures with all $0$'s and all $1$'s and decoding each, for a factor two increase in list size. (Of course list decoding from erasures is in general easier than this approach.)

To make all codewords far from both all $0$'s and all $1$'s, we can encode a string as $\ecc(x)||\overline{\ecc(x)}$ where $\overline{v}$ denotes the vector that flips every bit of vector $v$. Once again one can list decode both the first and second halves, so this preserves all the other desired properties up to constant factors.

With these observations, it suffices to meet the first three conditions for a block length $N_0 \le O(k/\poly(\epsilon))$. By now, there are several approaches to construct binary codes list-decodable in polynomial time up to an error fraction $\approx 1/2$. In view of the promised quasi-linear encoding and decoding time, we will use a concatenated scheme with an outer Reed-Solomon code. Specifically, we use a Reed-Solomon code over $\mathbb{F}_{2^m}$ of dimension $k$ and length $n_{\text{out}} = 2^m$, for a suitable choice of $m$. We will concatenate it with a binary linear inner code of dimension $m$, block length  $n_{\text{in}} \le O(m/\epsilon^{O(1)})$, that has relative distance more than $(1/2-\epsilon/10)$ and is list-decodable in $\poly(n_{\text{in}})$ time from a fraction $(1/2-\epsilon/10)$ of errors with lists of size $\poly(1/\epsilon)$. There are several options for such a code, for instance an algebraic-geometric code concatenated with Hadamard codes~\cite{GS-concat-stoc00}. The concatenated code can be used to encode $k$ bits into $N_0 := n_{\text{out}} \cdot n_{\text{in}}$ bits. Further, the encoding can be implemented in $\tilde{O}(N_0)$ time using FFT to encode the Reed-Solomon code and naive quadratic time decoding for each of the inner codes.

A simple approach to list decode the concatenated code would be list decode each inner block up to radius $(1/2-\epsilon/10)$, passing a set $S_i$ of possibilities for each outer Reed-Solomon codeword symbol, for positions $1 \le i \le 2^m$. If the total error fraction is at most $(1-\epsilon)/2$, at least $\alpha \ge \Omega(\epsilon)$ fraction of these sets will include the correct symbol. The inner decodings take $n_{\text{out}} \poly(m/\epsilon)$ time, which is $\tilde{O}_\epsilon(n_{\text{out}})$. The decoding can be completed by ``list-recovering" the outer Reed-Solomon code to find all codewords whose $i$'th symbol belongs to $S_i$ for at least $\alpha \cdot 2^m$ values of $i$. By the known list-recovery guarantees, this will succeed provided the rate of the Reed-Solomon code is at most $c_0 \cdot \epsilon^b$ for some absolute constant $c_0$ and large enough integer $b$~\cite{Sudan97}. So we can choose $m$ to be the smallest integer so that $2^m \cdot m \ge 2 k/(c_0 \epsilon^b)$. Further, the list recovery algorithm can be implemented in $\tilde{O}_\epsilon(n_{\text{out}})$ time~\cite{Alekhnovich05}. The distance of the concatenated code is at least $(1- c_0 \epsilon^b) (1/2-\epsilon/10) \ge (1-\epsilon)/2$.

It remains to bound the length $N_0$ of the code construction, which is $n_{\text{out}} \cdot n_{\text{in}} \le 2^m \cdot m/\poly(\epsilon)  \le O(k/\poly(\epsilon))$ by our choice of $m$. We have established all the required properties of our claimed coding scheme.
\end{proof}

\subsection{Multicolor Ramsey Theory}

We use a theorem about multicolor Ramsey numbers in our lower bound for the number of bits in Section~\ref{sec:lb-bits}.

\begin{theorem} \label{thm:ramsey} \cite{Erdos35}
    For positive integers $t$ and $\ell$, let $r(t; \ell)$ denote the $\ell$-color Ramsey number of $K_t$. This is the smallest integer $N$ such that every $\ell$-coloring of the edges of $K_N$ contains a monochromatic $K_t$. Then, it holds that $r(t; \ell) \leq \ell^{t\ell}$.
\end{theorem}

\subsection{Existence of Indices Partitioning Any Set}

We now state and prove a lemma that we use in both the error and erasure feedback ECC's. The lemma shows that for any values of $x_1, \dots, x_{L}\in \{0,1\}^k$, Bob can choose indices $i_1, \dots, i_{L} \in [k]$ such that $x_1[i_1, \dots, i_L], \dots, x_L[i_1, \dots, i_L]$ are all distinct tuples of $L$ bits.

\begin{lemma} \label{lem:index-list}
    For any $L$ strings $x_1, \dots, x_L \in \{ 0, 1 \}^k$ where $L \le k$, there exist $L$ indices $i_1, \dots, i_L \in [k]$ for which $x_1[i_1, \dots, i_L], \dots, x_L[i_1, \dots, i_L]$ are all distinct tuples of $L$ bits. Furthermore, these $L$ indices can be found in $O_L(k)$ time.
\end{lemma}

\begin{proof}
    We show the statement by induction on $L$. If $L=1$, choose $i_1$ arbitrarily (for example, let $i_1=1$). For the inductive step, assume we can find $i_1\ldots i_{L-1}$ that distinguish $x_1\ldots x_{L-1}$ in $O_L(k)$ time. Then, $x_L$ agrees with at most one  $x_j$ on the indices $i_1\ldots i_{L-1}$. Finding or disproving the existence of such an $x_j$ by evaluating the restriction of each element directly takes time $O_L(k)$. If $x_j$ exists, pick $i_L$ to be the first index on which $x_j$ and $x_L$ differ, which takes $O_L(k)$ time. If $x_j$ does not exist, let $i_L$ be the first index not already chosen, which takes time $O_L(k)$ to find.
    
    Then, the list of indices $i_1\ldots i_L$ distinguishes $x_1\ldots x_L$. As explained above, computing each of the $L$ indices successively takes time $O_L(k)$, so in total the algorithm takes time $O_L(k)$.
\end{proof}
\section{Optimal Feedback ECC for Erasures} \label{sec:erasure}
\subsection{Protocol}

Alice has a message $x \in \{ 0, 1 \}^k$, which she is trying to convey to Bob over a channel with an adversary who can erase bits. Bob may send $O_\epsilon(\log{k})$ bits of (noiseless) feedback over a constant number of rounds. In this section, we describe a protocol that achieves an erasure resilience of $1-\epsilon$ for any $\epsilon>0$.

For $\epsilon > 0$, let $\ECC = \ecc_\epsilon : \{ 0, 1 \}^k \rightarrow \{ 0, 1 \}^{n_0}$, where $n_0 = O_\epsilon(k)$, be the error-correcting code family from Theorem~\ref{thm:ecc} satisfying the erasure list-decoding property. Our feedback ECC is stated in Protocol~\ref{prot:logk-erasure}.

\protocol{Feedback ECC for Erasures}{logk-erasure}{
    The protocol consists of $T = (\log L_\epsilon + 1)/\epsilon$ rounds. In each round, Alice sends a message of length $n_0$. Then, Bob replies with a message of length $\beta = L_\epsilon\cdot(\log k + L_\epsilon)$. 
    
    \vspace{0.75em}
    {\bf Phase 1:}
    \vspace{0.25em}
    
    Alice begins by sending $\ECC(x)$ to Bob in every round. 
    \begin{itemize}
    \item 
        If Bob receives a message $m$ with at least $1 - \epsilon$ erasures, he replies by sending the string $0^\beta$ back to Alice.
    \item 
        Otherwise, if Bob receives a message $m$ with less than $1 - \epsilon$ erasures, he list-decodes $m$ to $L_\epsilon$ possibilities $x_1, \dots, x_{L_\epsilon} \in \{ 0, 1 \}^k$ (see Theorem~\ref{thm:ecc}).  He then picks $L_\epsilon$ indices $i_1, \dots, i_{L_\epsilon} \in [k]$ for which the values $x_1[i_1, \dots, i_{L_\epsilon}], \dots, x_{L_\epsilon}[i_1, \dots, i_{L_\epsilon}] \in \{ 0, 1 \}^{L_\epsilon}$ are all distinct (see Lemma~\ref{lem:index-list}). He sends $(i_1, i_2, \dots, i_{L_\epsilon}; x_1[i_1, \dots, i_{L_\epsilon}], \dots, x_{L_\epsilon}[i_1, \dots, i_{L_\epsilon}]) \in \{ 0, 1 \}^\beta$. Alice will receive a nonzero string, and both parties move on to Phase 2. 
    \end{itemize}
    
    {\bf Phase 2:}
    \vspace{0.25em}
    
    The last feedback Alice received upon entering Phase 2 was a string of the form $(i_1, \dots, i_{L_\epsilon}; y_1, \dots, y_{L_\epsilon})$, where $i_1, \dots, i_{L_\epsilon} \in [k]$ and $y_1, \dots, y_{L_\epsilon} \in \{ 0, 1 \}^{L_\epsilon}$, with the property that all $y_1, \dots, y_{L_\epsilon}$ are distinct. Alice determines a value $\gamma \in [L_\epsilon]$ for which $x[i_1, \dots, i_{L_\epsilon}] = y_\gamma$. Her goal is then to communicate $\gamma$ to Bob.
    
    \vspace{0.25em}
    
    To do this, Alice and Bob engage in the following procedure. Alice will send $\gamma$ to Bob bit by bit (viewing $\gamma$ as a $\log L_\epsilon$-bit string). Alice begins with $\nxt = 1$. Bob keeps track of a string $\hat{\gamma}$ originally set to the empty string. 
    \begin{itemize}
    \item 
        Every round, Alice sends $(\gamma[\nxt])^{n_0}$, where if $\nxt > \log L_\epsilon$ she simply sends $0^{n_0}$. 
    \item 
        If Bob receives a message $m$ where not every bit is erased, he appends one of the unerased bits to $\hat{\gamma}$. Finally, he sends $1^\beta$ to Alice. If Alice receives $1^\beta$, she sets $\nxt \gets \nxt + 1$.
        
    \item 
        If every bit of $m$ is erased, he does not alter $\hat{\gamma}$ and sends $0^\beta$ to Alice.
    \end{itemize}
    
    At the end of the protocol, Bob takes $\hat{\gamma}' = \hat{\gamma}[1:\log L_\epsilon]$. He outputs $x_{\hat{\gamma}'}$. 
}
\subsection{Analysis}

\begin{theorem}
    Protocol~\ref{prot:logk-erasure} is a feedback $\ECC$ with the following properties:
    \begin{itemize}
        \item {\bf Erasure Resilience:} For any erasure of $\le 1 - 2\epsilon$ fraction of Alice's communication, Bob outputs $x$. 
        \item {\bf Forward Communication:} Alice sends a total of $O_\epsilon(k)$ bits.
        \item {\bf Feedback Complexity:} There are $O_\epsilon(1)$ rounds and a total of $O_\epsilon(\log k)$ bits of feedback.
        \item {\bf Efficiency:} Alice and Bob both run in time $\tilde{O}_\epsilon(k)$. 
    \end{itemize}
\end{theorem}

\begin{proof}
    To see the claim about the feedback complexity, note that our protocol runs in $T = O_\epsilon(1)$ rounds, and each round Bob sends $O_\epsilon(k)$ bits of feedback. As for efficiency: For Alice, note that computing $\ECC(x)$ takes $\tilde{O}_\epsilon(n_0) = \tilde{O}_\epsilon(k)$ time by Theorem~\ref{thm:ecc}, and determining $\gamma$ from the list $(i_1, \dots, i_{L_\epsilon}; y_1, \dots, y_{L_\epsilon})$ takes $O_\epsilon(L_\epsilon \cdot k) = O_\epsilon(k)$ time. Finally, Alice sending Bob the bits of $\gamma$ takes $O_\epsilon(1)$ time per round. This gives a total of $\tilde{O}_\epsilon(k)$ running time for Alice. Meanwhile, for Bob, list-decoding takes $\tilde{O}_\epsilon(n_0) = \tilde{O}(k)$ time by Theorem~\ref{thm:ecc}, and computing the feedback $(i_1, \dots, i_{L_\epsilon}; x_1[i_1, \dots, i_{L_\epsilon}], \dots, x_{L_\epsilon}[i_1, \dots, i_{L_\epsilon}])$ takes $O_{L_\epsilon}(k) = O_\epsilon(k)$ time by Lemma~\ref{lem:index-list}. Finally, receiving $\hat{\gamma}$ bit by bit takes $O_\epsilon(1)$ time per round, and determining $x_{\hat{\gamma}'}$ takes $O_\epsilon(k)$ time. In total, this gives that Bob runs in $\tilde{O}_\epsilon(k)$ time as well. 
    
    We focus on proving erasure resilience to $\le 1 - 2\epsilon$ of Alice's communication.

    Given erasure of $\le 1 - 2\epsilon$ of Alice's communicated bits, suppose that in $\alpha$ of the rounds, less than $1 - \epsilon$ of Alice's communication is erased, and in the other $T - \alpha$ rounds, at least $1 - \epsilon$ fraction of Alice's message is erased. Then 
    \begin{align*}
        (T - \alpha) \cdot (1-\epsilon) &\le (1 - 2\epsilon) \cdot T \\
        \implies T - \alpha &\le (1 - \epsilon) \cdot T \\
        \implies \alpha &\ge \epsilon T = \log L_\epsilon + 1. 
    \end{align*}
    In the first of the $\alpha$ rounds where $< 1 - \epsilon$ of Alice's communication is erased, Bob must be able to list decode Alice's message to $L_\epsilon$ options $x_1, \dots, x_{L_\epsilon}$, one of which is guaranteed to be $x$. Based on his feedback, which is nonzero, both parties move to Phase 2.
    
    In Phase 2, there are $\alpha - 1 \ge \log L_\epsilon$ rounds in which Bob receives a nonzero message from Alice. For the first $\log L_\epsilon$ of such rounds, Bob receives the next bit of $\gamma$ from Alice: since his feedback is uncorrupted, whenever he receives a bit from Alice, Alice immediately moves on to sending the next bit of $\gamma$, so that Bob receives the $\log L_\epsilon$ bits of $\gamma$ in order. Thus, $\hat{\gamma}' = \gamma$.
    
    Finally, we claim that $x_\gamma = x$. To see this, recall that by the erasure list-decoding guarantee of Theorem~\ref{thm:ecc}, $x$ must be among $x_1, \dots, x_{L_\epsilon}$. Furthermore, since $x_1, \dots, x_{L_\epsilon}$ have distinct values on the indices $i_1, \dots, i_{L_\epsilon}$ by Lemma~\ref{lem:index-list}, there is a unique value of $\gamma \in [L_\epsilon]$ for which $x[i_1, \dots, i_{L_\epsilon}] = y_\gamma$. 
    
    In conclusion, Bob will have determined a list $x_1, \dots, x_{L_\epsilon}$ containing $x$, along with the value $\gamma \in [L_\epsilon]$ such that $x = x_\gamma$, assuming that at most $1 - 2\epsilon$ of Alice's communication is erased. This allows him to output $x$ correctly.
\end{proof}

\section{Optimal Feedback ECC for Errors} \label{sec:error}
\subsection{Ingredients}

We begin by defining two ingredients---partitions of the message space and error-correcting codes encoding the message space---that the parties will use repeatedly in the protocol.

\subsubsection{Partitions}

We define a way for Bob to specify partitions of the set of Alice's possible input using only $O_\epsilon(\log{k})$ bits. The function $\partition$ takes in as input a sequence $X$ of constantly many elements of $\{ 0, 1 \}^k$ and outputs a short description (of size $O_\epsilon(\log k)$) of a partition of the input space such that the specified elements are all in separate sets of the partition. The number of sets in the partition is $|X|$.

\begin{definition}[$\partition$ and $\partdec$]\label{def:partition}
    Let $\partition(X)$ be a function that takes in a sequence $X:= \{x_i\}_{i=1}^{|X|}$ with $x_i\in \{0,1\}^k$ and outputs a string in $\{0,1\}^{|X|\log{k}+|X|^2}$ describing a partition of $\{ 0, 1 \}^k$ into $|X|$ sets. Let $\partdec(s,x)$ be a dual function that takes a string $s\in \{0,1\}^*$ and an element $x\in\{0,1\}^k$ and outputs an integer (corresponding to which set of the partition described by $s$ that $x$ is in). These functions satisfy that for all $x_i \in X$, 
    \[
        \partdec(\partition(X), x_i)=i.
    \]
\end{definition}

\begin{claim}
    The dual functions $\partition$ and $\partdec$ with the properties in Definition~\ref{def:partition} exist. Moreover, the function $\partition(X)$ can be computed in $O_{|X|}(k)$ time, and $\partdec(s,x)$ can be computed in $O(|s|,k)$ time.
\end{claim}

\begin{proof}
    Let $k\in \bbN$ and a sequence $X:= \{x_i\}_{i=1}^{|X|}$ with $x_i\in \{0,1\}^k$. By Lemma~\ref{lem:index-list}, there is a set of indices $I\subseteq [k]$ with $|I|=|X|$ such that for all $i \not= j\in [|X|]$, it holds that $x_i[I]\neq x_j[I]$. The description $\partition(X)$ is simply the set $I$ and the values $x_1[I], \dots, x_{|X|}[I]$. This description has size $|X| \cdot \log k + |X|^2$ and can be computed in $O_{|X|}(k)$ time by Lemma~\ref{lem:index-list}.
    
    The function $\partdec(s\in \{0,1\}^*,x)$ first recovers $I\subseteq [k]$ and the $|X|$ evaluations $x_1[I], \dots, x_{|X|}[I]$ from the string $s$ (if $s$ does not describe these values, the function can just output $1$). Then it outputs $i$ for which $x[I] = x_i[I]$, and if no such $i$ exists, then it outputs $1$. For all $i\in [|X|]$, it holds that $\partdec(\partition(X),x_i)=i$ so Definition~\ref{def:partition} is satisfied.
\end{proof}

\subsubsection{Error-Correcting Codes}

In every round of the protocol, Alice will send one of two encodings of $x$. The first is the error-correcting code from Theorem~\ref{thm:ecc}, with the choice of block length $3N_0$.

\begin{definition}\label{def:ecc}
    For fixed $k \in \bbN, \epsilon > 0$, we define the error-correcting code
    \[
        \ECC : \{ 0, 1 \}^k \rightarrow \{ 0, 1 \}^{M(k)}
    \]
    as in Definition~\ref{thm:ecc} with $M(k) = 3 \cdot N_0(k)$.
\end{definition}

We now describe a second ECC that Alice uses to send Bob her input $x$. This second ECC has special properties so that some pairs of possible inputs are separated by significantly more than $\frac12$ while others are separated by significantly less. In particular, Alice bases her encoding on a partition of her input space into $D$ sets that Alice had sent her (as $s \in \{ 0, 1 \}^{D \log k + D^2}$), and a further partitioning of these sets into three sets, specified by a function $f : [D] \rightarrow \{1,2,3\}$. Inputs $x$ that fall in different sets in this new coarser partition are separated by $\frac23$ and elements that fall in the same set are separated by less: specifically, elements in the first set of this partition are separated by $\frac13$ and elements in the second or third sets are separated by $0$. The code's parameters are formalized below.

\begin{definition}[{$\ECC[s,f]$}]\label{def:ecc3}
    For all $k$ and fixed $\epsilon > 0$, we define an explicit ECC for each string $s\in \{0,1\}^{D\log{k}+D^2}$ for some $D$ and function $f:[D]\to \{1,2,3\}$. The code is denoted
    \[
        \ECC[s,f] : \{0,1\}^k \rightarrow \{ 0, 1 \}^{M(k)}
    \] and has the following properties.
    
    \begin{enumerate}[label=(\alph*)]
    \item 
        $M = M(k)=3N_0(k)$ is the same block length as the code in Definition~\ref{def:ecc}. 
    \item \label{item:ecc3-encodable}
        The code is efficiently encodable. That is, for any $x, s, f$, one can compute $\ECC[s,f](x)$ in $\tilde{O}_{\epsilon}(D, k)$ time.
    \item \label{item:2,3-equal}
        For any $x,y \in \{ 0, 1 \}^k$ such that $f(\partdec(s,x))=f(\partdec(s,y)) \in \{2,3\}$, we have
        \[
            \ECC[s,f](x) = \ECC[s,f](y)
        \]
    \item \label{item:ecc-1/3}
        For any $x \not= y \in \{ 0, 1 \}^k$ such that $f(\partdec(s,x))=f(\partdec(s,y))=1$, we have
        \[
            \Delta \big( \ECC[s,f](x), \ECC[s,f](y) \big) \ge \left( \frac13 - \epsilon \right) \cdot M, \numberthis \label{eqn:ecc-1/3}
        \]
    \item \label{item:ecc-2/3}
        For any $x,y$ such that $f(\partdec(s, x)) \not= f(\partdec(s, y))$, it holds that
        \[
            \Delta \big( \ECC[s,f](x), \ECC[s,f](y) \big) \ge \left( \frac23 - \epsilon \right) \cdot M, \numberthis \label{eqn:ecc-0}
        \]
    \item \label{item:ecc-1/3-list-decoding}
        For any message $m\in \{0,1\}^M$, the number of values of $x \in \{ 0, 1 \}^k$ with $f(\partdec(s, x)) = 1$ satisfying 
        \[
            \Delta \big( \ECC[s,f](x), m \big) \le \left( \frac13 - 2\epsilon \right) \cdot M, \numberthis \label{eqn:ecc-list-1/3}
        \]
        is at most the constant $L_\epsilon$ (see Theorem~\ref{thm:ecc}). Moreover, there is an algorithm to compute these values in $\tilde{O}_\epsilon(k)$ time.
    \end{enumerate}
\end{definition}

\begin{claim} \label{claim:ecc3}
    The code $\ECC[s,f]$ as defined in Definition~\ref{def:ecc3} exists. 
\end{claim}

\begin{proof}
    Take an error-correcting code $\ecc : \{ 0, 1 \}^k \rightarrow \{ 0, 1 \}^{2N_0}$ with relative distance $\frac12 - \epsilon$ as in Theorem~\ref{thm:ecc}. We define $\ECC[s,f](x)$ as follows:
    \begin{align*}
        \ECC[s,f](x) = \begin{cases}
            0^{N_0} || \ecc(x) & f(\partdec(s, x)) = 1 \\
            1^{N_0} || 0^{2N_0} & f(\partdec(s, x)) = 2 \\
            1^{N_0} || 1^{2N_0} & f(\partdec(s, x)) = 3.
        \end{cases}
    \end{align*}
    
    Note that $M = |\ECC[s,f](x)| = 3 N_0$ is the same as in Definition~\ref{def:ecc}, and that $\ecc$ has block length $2N_0$. Properties~\ref{item:ecc3-encodable} and~\ref{item:2,3-equal} are clear from the definition.
    
    As for property~\ref{item:ecc-1/3}, note that 
    \begin{align*}
        \Delta(\ECC[s,f](x), \ECC[s,f](y)) 
        &= \Delta(0^{N_0} || \ecc(x), 0^{N_0} || \ecc(y)) \\
        &\ge \left( \frac12 - \epsilon \right) \cdot N_0 \\
        &\ge \left( \frac13 - \epsilon \right) \cdot M.
    \end{align*}
    And for property~\ref{item:ecc-2/3}, we see that 
    \begin{align*}
        \Delta(0^{N_0} || \ecc(x) , 1^{N_0} || 0^{2N_0)}) &\ge N_0 + \left( \frac12 - \epsilon \right) \cdot 2N_0 \ge \left( \frac23 - \epsilon \right) \cdot M \\
        \Delta(0^{N_0} || \ecc(x), 1^{N_0} || 1^{2N_0}) &\ge N_0 + \left( \frac12 - \epsilon \right) \cdot 2N_0 \ge \left( \frac23 - \epsilon \right) \cdot M \\
        \Delta(1^{N_0} || 0^{2N_0}, 1^{N_0} || 1^{2N_0}) &= \frac23 \cdot M.
    \end{align*}
    Finally, for property~\ref{item:ecc-1/3-list-decoding}, note that if $f(\partdec(s,x)) = 1$,
    \[
        \Delta(\ECC[s,f](x), m) \le \left( \frac13 - 2\epsilon \right) \cdot M \implies \Delta(\ecc(x), m[N_0+1:3N_0]) \le \left( \frac12 - 2\epsilon \right) \cdot 2N_0,
    \]
    so by Theorem~\ref{thm:ecc} there are at most $L_\epsilon$ values of $x \in \{ 0, 1 \}^k$ with $f(\partdec(s,x))=1$ satisfying $\Delta(\ECC[s,f](x), m) \le \left( \frac13 - 2\epsilon \right) \cdot M$. These can be found via the list decoding algorithm that runs in $\tilde{O}_\epsilon(k)$ time, and checked that the resulting codewords are in fact within $\left( \frac13 - 2\epsilon \right) \cdot M$ of $m$. 
\end{proof}

\subsection{Definition of a Coin Game}\label{sec:coin-game-def}

Before we can state the formal protocol, we formally define a coin game. The coin game is a data structure that helps Bob store, access, and update values associated with some set of \emph{coins}. Bob will initialize instances of the coin game at the start of the protocol, and in each chunk where he opts to play the coin minigame. Below, we define the coin game data structure, and show that the required access and update operations can be done efficiently.

\datastruct{Coin Game }{coin-game}{
    A \emph{coin game} $\cG$ is a data structure for any set $S$\footnote{Though we do not require a particular domain for $S$, we do require a few specific properties. The elements can be specified and accessed in $\log{|S|}$ space/time, and that $S$ comes with an ordering where two elements can be compared efficiently, and a way to generate the smallest $m\in \bbN$ elements in order in $O(m\log{|S|})$ time. In our protocol, the set $S$ is either $\{0,1\}^N$ or $[N]$ for $N\in \bbN$, which clearly have these properties.} and a size parameter $K\in \bbN$ representing the maximum size of a single update. For each element (``coin'') $y \in S$ and for every time step $t \in \bbN_{\ge 0}$, the data structure maintains a value $\pos_\cG(y) \in \bbZ_{\ge 0}$. We will denote the value of $\pos_\cG(y)$ at time $t$ by $\pos_\cG(y; t)$. Initially, at time $t = 0$, it holds that $\pos_\cG(y; 0) = 0$ for all $y \in S$. At every time $t > 0$, the values $\pos_\cG(y)$ in the data structure are updated. We list the values that can be queried from this data structure, and the updates that can be made to it. Each query and update takes time $O_{t,K}(\log{|S|})$, and the total storage is at most $O_{t,K}(\log{|S|})$ as well. 
    
    \vspace{0.4em}
    
    \noindent{\bf Queries:}
    \begin{itemize}
        \item For any coin $y \in S$, the query $\pos_\cG(y;t)$ outputs a value assigned to $y$. 
        \item For any $i\leq |tK|$, the query $x^t(i)$ outputs the coin $\in S$ with the $i$'th smallest value of $\pos_\cG(y; t)$ over any $y\in S$. Ties are broken by the sorting on $S$.
        \item For any $i\leq |tK|$, the query $\posx_\cG(i;t)$ returns the value $\pos_\cG(x^t(i);t)$. 
    \end{itemize} 
    
    \noindent{\bf Update:}
    \begin{itemize}
        \item At each time $t>0$, the function $\update_\cG(\cL,v)=0$ takes in a list $\cL$ with $|\cL|\leq K$ of index-value pairs $(y,v)$ with $y\in S$ and $v\in \bbN$ and no repeated values of $y$, as well as an extra value $V\in \bbN$. For each coin-value pair $(y,v)\in \cL$, it should hold that $\pos_\cG(y;t)=\pos_\cG(y;t-1)+v$. Moreover, for $y$ not contained in the list $\cL$, it holds that $\pos_\cG(y;t)=\pos_\cG(y;t-1)+V$. We mention that the values $v$ and $V$ are required to be at most $O_K(\log{|S|})$ bits.
    \end{itemize}
}

Now we show that a coin game data structure exists.

\begin{claim}\label{claim:coin-game}
    There exists a coin game data structure, as in Data Structure~\ref{def:coin-game}. 
\end{claim}

\begin{proof}
    We explicitly build a coin game $\cG$ with coin set $S$ and maximum update size $K$. At every time $t\geq 0$, the data structure stores two things. The first is a list $\cP(t)$ of length $\leq tK$ of pairs $(y\in S, p_y \in \bbN)$ with no redundant coins $y$. The second is a single value $U(t) \in \bbN$.
    
    The list $\cP(0)$ is initialized to the empty list and $U(0)$ is initialized to $0$. 
    
    The query $\pos_\cG(y;t)$ for $y\in S$ checks if there exists a pair $(y,p_y)\in \cP(t)$. If so, it returns $p_y$. If not, it returns $U(t)$. 
    
    The query $x_\cG^t(i)$ for some $i\in [tK]$, creates a new list $\cP'$ as follows. The list $\cP'$ consists of all the pairs in $\cP(t)$, as well as the pairs $(y,U(t))$ for all $y \in S(tK)$ for which there is not a pair $(y, p_y\in \bbN)$ in $\cP(t)$, where $S(tK)$ represents the $tK$ smallest elements of $S$. Then, it sorts $\cP'$ according to the second element of each pair (the value) breaking ties according to the ordering on $S$ (taking time $O(tK)$). Then, if $(y,v)$ is the $i$'th pair in this list, it returns $y$. This procedure returns the coin $y$ with the $i$'th smallest value of $\pos_\cG(y;t)$ without incorporating all $(y\in S, \pos_\cG(y; t))$ into the list, because only the smallest $tK$ values can be queried. The query takes time $O_{t,K}(\log{|S|})$. 
    
    The query $\posx_\cG(i;t)$ is just a definition, so we do not need to prove anything about it.
    
    When at time $t$ the coin game receives $\update_\cG(\cL,V)$ with $\cL$ of index-value pairs $(y,v)$ with $y\in S$ and $v\in \bbN$ with $|\cL|\leq K$ and a value $V\in \bbN$, we create $\cP(t)$ and $U(t)$ as follows. 
    \begin{itemize}
    \item 
        For every $(y, p_y) \in \cP(t-1)$, if a pair $(y, v)$ appears in $\cL$, add the pair $(y, p_y + v)$ to $\cP(t)$. Else, if the pair $(y, v)$ does not appear in $\cL$, add instead the pair $(y, p_y + V)$. 
    \item 
        For all $(y, v)$ in $\cL$ where which there does not exist $(y, p_y) \in \cP(t-1)$, add the pair $(y, U(t-1) + v)$ to $\cP(t)$.
    \item 
        Define $U(t) = U(t-1) + V$.
    \end{itemize}
    Then, it holds that $\pos_\cG(y;t)=\pos_\cG(y;t-1)+v$ if $(y,v) \in \cL$ and $\pos_\cG(y;t)=\pos_\cG(y;t-1)+V = U(t-1)  +V$ otherwise, which is what we wanted. Moreover, there continue to be no redundant pairs, and the process takes time $O_{t,K}(\log{|S|})$.
    
\end{proof}

\subsection{Statement of the Protocol}

\protocol{Feedback ECC for Errors}{error}{
    Fix $\epsilon<0.01$\footnote{We assume $\frac1\epsilon \in \bbN$.} and $k\in \bbN$. Alice has a message in $\{0,1\}^k$ that she wishes to convey to Bob. The protocol consists of $C=\frac{1}{\epsilon}$ \emph{chunks} of $R=\frac{100 \ln L_\epsilon}{\epsilon^3}$ \emph{rounds}. In each round, Bob sends a message of length $2CRL_\epsilon$ to Alice as feedback, and in response Alice sends a message of length $M=M(k)$ to Bob. Before the start of each chunk, Bob additionally sends a message of length $CRL_\epsilon\log{k}+(CRL_\epsilon)^2+2$ called his \emph{prompt}.

    At the start of the protocol, Bob initializes a coin game $\main$ with $S_\main:=\{0,1\}^k$ and the selected constant $K_\main:=CRL_\epsilon$. At the end of each chunk, Bob updates the state of $\main$ using the function $\update_\main$, so that, for example, the position of coin $y$ after chunk $t$ is accessible as $\pos_{\main}(y; t)$. At the end of the protocol, Bob outputs $x_\main^{C}(1)$.

    We begin by describing the prompt Bob sends at the start of each chunk $t$.
    
    \paragraph{Bob's prompt:} 
    
    The first two bits of Bob's prompt tell Alice which case below to enter for the duration of the chunk, and the other bits provide her with other necessary information. He determines his prompt as follows:
    
    \begin{itemize}
    \item 
        If $\posx_\main(3;t-1)>\frac13 \cdot CRM$, Bob computes the first index $a\in[k]$ where $x_\main^{t-1}(1)[a]\neq x_\main^{t-1}(2)[a]$ differ. Then, he sends $00a0^*$\footnote{In general, when we write $0^*$, we mean that the remainder of Bob's feedback is padded with $0$'s until it is the correct length.} (where $a$ is written in binary), so that Alice knows to enter Case~\ref{case:1} below and also knows the index $a$.
    \item 
        Othewise, if $\posx_\main(3;t-1)>\posx_\main((t-1) RL_\epsilon+1;t-1)-\epsilon CRM$, then Bob sends $010^*$ so Alice knows to enter Case~\ref{case:2} below.
    \item
        Otherwise, he sends $10 || \partition(x_\main^{t-1}(1),\ldots, x_\main^{t-1}((t-1) RL_\epsilon))$.
    \end{itemize}
    
    \paragraph{The chunk:}
    
    The prompt dictates the parties' behavior throughout the rest of chunk $t$.

    \begin{caseofb}
    \caseb{Bob's prompt was $00a0^*$ for some index $a\in \{ 0, 1 \}^{\log k} = [k]$:\footnote{We remark that this case was Case 3 in the technical overview, and that Cases 2 and 3 here are Cases 1 and 2 in the technical overview.}}{\label{case:1}
        \begin{itemize}
        \item {\bf Alice's messages and Bob's feedback:} 
            Alice sends the bit $x[a]$ as all her bits in the chunk, ignoring Bob's feedback entirely. Bob sends the feedback $0^*$ at the start of every round. 
        \item {\bf Bob's update to $\main$: } 
            At the end of the chunk, Bob has received a message $m \in \{ 0, 1 \}^{RM}$ from Alice. Let $b_1$ be the number of bits in $m$ not equal to $x_\main^{t-1}(1)[a]$ and $b_2$ be the number of bits in $m$ not equal to $x_\main^{t-1}(2)[a]$. Bob creates the list
            \[
                \cL := [(x_\main^{t-1}(1),b_1),(x_\main^{t-1}(2),b_2)].
            \]
            and performs $\update_\main(\cL,0)$.
        \end{itemize}
    }
    \caseb{Bob's prompt was $010^*$:} {\label{case:2}
        \begin{itemize}
        \item {\bf Alice's messages and Bob's feedback:} 
            Alice sends $\ECC(x) \in \{ 0, 1 \}^{RM}$ using all $R$ rounds in the chunk. More precisely, in round $\tau$ of the chunk, she sends $\ECC(x)[(\tau-1)M+1:\tau M]$, and ignores Bob's feedback. Bob sends the feedback $0^*$ at the start of every round. 
        \item {\bf Bob's update to $\main$: } 
            After the entire chunk, Bob has received the message $m \in \{ 0, 1 \}^M$ as length $RM$ concatenation of what Alice sent throughout the chunk. He updates his state as follows. First, he computes (via list-decoding, see Theorem~\ref{thm:ecc}) the list $\cL$ of size $\le L_\epsilon$ of all $(y\in \{0,1\}^k,\Delta(\ECC(y),m))$ where $y$ is such that $\Delta(\ECC(y),m)\leq \left(\frac12-2\epsilon\right)\cdot RM$. Then, he performs $\update_\main(\cL,\left(\frac12-2\epsilon\right) \cdot RM)$.
        \end{itemize}
    }
    \caseb{Bob's prompt was $10s10^*$ for $s \in \{0,1\}^*$:} {\label{case:3}
        Alice and Bob play the coin minigame for the rest of the chunk. Bob instantiates a new coin game $\mini$, which we call the coin minigame, on the set $S_\mini:=[(t-1)RL_\epsilon]$ and the constant $K_\mini:=(t-1)RL_\epsilon$. For each round $\tau \in [R]$, we describe Bob's feedback, Alice's message, and Bob's update to $\mini$. We also define Bob's update to $\main$ at the end of the chunk.
    
        \begin{itemize}
        \item {\bf Bob's feedback:} 
            Bob defines $f^\tau: [(t-1)RL_\epsilon] \to \{1,2,3\}$ such that $f^\tau(x^{\tau-1}_\mini(1))=1$, $f^\tau(x^{\tau-1}_\mini(2,4,6\ldots))=2$, and $f^\tau(x^{\tau-1}_\mini(3,5,7\ldots))=3$. He sends this function $f^\tau$ to Alice, by listing the output of $f^\tau$ on each of $[(t-1)RL_\epsilon]$ in order (which takes at most $2CRL_\epsilon$ bits).
            \item {\bf Alice's Message:} Alice receives the function $f^\tau$ and sends the message $\ECC[s,f^\tau](x)$.
        \item {\bf Bob's update to coin minigame:} 
            In each round, Bob receives a corrupted message $m\in \{0,1\}^M$ from Alice. Bob computes a list $\Lambda_\tau$ with $|\Lambda_\tau|\leq L_\epsilon$ of $y\in \{0,1\}^k$ such that $f^\tau(\partdec(s,y))=1$ and $\Delta(\ECC[s,f^\tau](y),m))\leq \left(\frac13-2\epsilon\right)M$ using the list-decoding property of $\ECC[s,f^\tau]$ in Definition~\ref{def:ecc3}. For each $i\in [(t-1)RL_\epsilon]$, Bob adds $(i,v^\tau(i))$ to a list $\cL$, where 
            \[ 
                v^\tau(i)=
                \begin{cases}
                    \min_{y\in\Lambda_\tau}\Delta(\ECC[s,f^\tau](y),m)  & f^\tau(i)=1 ~\text{and}~ \Lambda_\tau \not= \emptyset\\
                    \left( \frac23 - \epsilon \right) \cdot M - \min \{ d_2^\tau, d_3^\tau \} & f^\tau(i)=1 ~\text{and}~ \Lambda_\tau = \emptyset\\
                    d_{f^\tau(i)}^\tau & f^\tau(i)\neq 1,
                \end{cases}
            \]
            where $d_2^\tau := \Delta(\ECC[s, f^\tau](x_\main^{t-1}(2)), m)$ and $d_3^\tau := \Delta(\ECC[s, f^\tau](x_\main^{t-1}(3)), m)$.
            
            He then performs $\update_\mini(\cL,M)$.
            
        \item {\bf Bob's update to main coin game:} 
            At the end of the chunk, Bob updates $\main$. He creates a list $\cL$ of size at most $tRL_\epsilon$ as follows. For each $i\in[(t-1)RL_\epsilon]$, Bob adds $(x_\main^{t-1}(i),\pos_\mini(i;R))$ to $\cL$. For each $y\in \bigcup_\tau \Lambda_\tau$, Bob adds $(y,\pos_\mini(\partdec(s,y);R))$ to $\cL$. He deletes any redundant pairs in $\cL$. Then, he performs $\update_\main(\cL, \left(\frac13-4\epsilon\right)RM)$.
        \end{itemize}
    
    }
    \end{caseofb}

}
\subsection{Analysis}

We now state our main error resilience theorem.

\begin{theorem} \label{thm:error-feedback-ecc}
    Protocol~\ref{prot:error} is resilient to a $\left(\frac13-6\epsilon\right)$-fraction of adversarial corruption. Moreover, the communication complexity for Alice is $O_\epsilon(k)$ and the feedback complexity for Bob is $O_\epsilon(\log{k})$ and he speaks in $O_\epsilon(1)$ rounds. The computational complexity for both parties is $\tilde{O}_\epsilon(k)$.
\end{theorem}

We split the proof of Theorem~\ref{thm:error-feedback-ecc} into three parts. In Section~\ref{sec:analysis-complexity}, we discuss the communication and computational complexity of Protocol~\ref{prot:error}. Next, we establish in Section~\ref{sec:analysis-coingame} that at every point in time, the positions of coins in the coin game $\main$ is a lower bound on the amount of errors the adversary has actually caused if Alice's input were a given coin, and also some lower bounds on how much the coins move in individual chunks. A crucial part of this is analysis of the coin minigame, which we defer to Section~\ref{sec:minigame-analysis}. Finally, in Section~\ref{sec:analysis-error}, we put together these claims and analyze the error resilience of Protocol~\ref{prot:error}.


\subsubsection{Communication and Computational Complexity} \label{sec:analysis-complexity}

\begin{lemma}
    In Protocol~\ref{prot:error}, the communication complexity for Alice is $O_\epsilon(k)$ and the communication complexity for Bob is $O_\epsilon(\log{k})$. The computational complexity for both parties is $\tilde{O}_\epsilon(k)$.
\end{lemma}

\begin{proof}
    First, we show the communication complexity claims. Alice sends $CR$ rounds of $M$ bits. In total, the number of bits sent is 
    \[
        CRM = O_\epsilon(1)\cdot M = O_\epsilon(k).
    \]
    Next, we show the communication complexity for Bob is $O_\epsilon(\log{k})$. The number of bits sends per prompt (every chunk) is $CRL_\epsilon\log{k}+(CRL_\epsilon)^2+2$ (we recall that he pads his message with $10^*$ to keep message sizes the same), and the number of bits he sends per round is $2CRL_\epsilon$. In total, the number of bits is 
    \[
        (CRL_\epsilon\log{k}+(CRL_\epsilon)^2+2) \cdot C
        + 2CRL_\epsilon \cdot CR
        = O_\epsilon(\log{k}) + O_\epsilon(1) = O_\epsilon(\log{k}).
    \]
    
    Finally, we show that the computational complexity for both parties is $\tilde{O}_\epsilon(k)$. To do this, we go through the protocol in steps and show that individual steps have computational complexity $\tilde{O}_\epsilon(k)$.
    
    First, Bob initializes the coin game $\main$ which takes time $O_\epsilon(k)$. In general, all the operations to $\main$ take time $O_{CRL_\epsilon}(k)=O_\epsilon(k)$, so we do not discuss them individually.
    
    Next, he decides what prompt to send at the start of each chunk $t$. After querying $\main$, he might need to compute $\partition$ which takes time $O_{(t-1)RL_\epsilon}(k)=O_{\epsilon}(k)$.
    
    In Case~\ref{case:1} of the protocol, Alice can determine and send her bit in $O_\epsilon(k)$ time, and Bob determining the number of $0$'s and $1$'s and updating $\main$ takes $O_\epsilon(k)$ time as well.
    
    In Case~\ref{case:2}, Alice needs to encode $\ECC(x)$ which we know by Definition~\ref{def:ecc} takes $\tilde{O}_\epsilon(k)$ time, and Bob must decode this into a list of size $L_\epsilon$ which also by Definition~\ref{def:ecc} takes $\tilde{O}_\epsilon(k)$ time. Then he updates the coin game.
    
    In Case~\ref{case:3}, Alice and Bob play the coin minigame initializing $\mini$, whose operations all have complexity $O_\epsilon(1)$, because the size of the set $S_\mini$ and $K_\mini$ are both at most $CRL_\epsilon$. In each round, Bob computes his feedback, which requires $(t-1)RL_\epsilon$ queries to the minigame, and then describing the input-output pairs of a function $f^\tau : [(t-1)RL_\epsilon] \rightarrow \{ 1, 2, 3 \}$ which he sends to Alice (thus also taking time $O_\epsilon(1)$). Alice computes the message $\ECC[s,f^\tau](x)$, where $s$ and $f^\tau$ are defined by Bob's prompt and round feedback, which takes time $\tilde{O}_\epsilon(k,|s|)=\tilde{O}_\epsilon(k)$. Bob's update to $\mini$ requires him to list-decode $\ECC[s,f_\tau]$ which he can do in time $\tilde{O}_\epsilon(k)$, and otherwise just perform a series of additions to a constant sized list $\cL$. Finally, his update to the main coin game at the end of the chunk only involves a constant sized list and a single update to $\main$, and therefore takes time $O_\epsilon(k)$.
\end{proof}

\subsubsection{Coin Updates in Individual Chunks} \label{sec:analysis-coingame}

Now, we proceed to a series of lemmas that will help us prove the protocol is resilient to $\left(\frac13-O(\epsilon)\right)$ fraction of adversarial corruption. These lemmas provide lower bounds on how much the coins update in each individual chunk based on whether Alice and Bob entered Case~\ref{case:1}, Case~\ref{case:2}, and Case~\ref{case:3}. Moreover, we show for any coin $y$ and chunk $t$, that $\pos_\main(y; t)$ is always a lower bound on the amount of corruption the adversary caused if Alice's input were actually $y$.

\begin{lemma} \label{lem:single-bit}
    In any chunk $t\in [C]$ where Alice and Bob entered Case~\ref{case:1}:
    
    \begin{enumerate}[label={(\roman*)}]
    \item At the end of the chunk, for any coin $y \in \{ 0, 1 \}^k$, 
    \[
        \pos_\main(y; t) - \pos_\main(y; t-1) 
    \]
    is a lower bound on the amount of corruption the adversary has caused if Alice's input were actually $y$.
    \item As long as $\posx_\main(2;t-1)<\left(\frac13-\epsilon\right)\cdot RM$, we have
    \[
        \posx_\main(1;t)+\posx_\main(2;t)-(\posx_\main(1;t-1)+\posx_\main(2;t-1)) \geq RM.
    \]
    \end{enumerate}
\end{lemma}

\begin{proof}
    At the end of the chunk $t$, let us say Bob has received the length $RM$ message $m$. If Alice's true value was $y = x_\main^{t-1}(1)$, then she would have sent a message consisting of all the bit $y[a]$, and $\pos_\main(y; t) - \pos_\main(y; t-1)$ is the number of messages that were not equal to this bit. The same applies for $x_\main^{t-1}(2)$. In all other cases, $\pos_\main(y; t) - \pos_\main(y; t-1)$ is $0$ which is always a lower bound on the amount of corruption Eve could have caused in a chunk.
    
    For the second claim, we see that at the start of the chunk, it held that $\posx_\main(3;t-1)>\frac13\cdot CRM$. Therefore, assuming that $\posx_\main(2;t-1)<\left(\frac13-\epsilon\right)\cdot CRM = \frac13 \cdot CRM - RM$ means that the first two coins will continue to be the first two coins after the chunk (since no coin can update by more than $RM$ in a chunk). Both of these coins have their positions updated explicitly at the end of the chunk, by the number of bits of the received message $m$ that were not $x_\main^{t-1}(1)[a]$ and not $x_\main^{t-1}(2)[a]$ respectively. Since $x_\main^{t-1}(1)[a]\neq x_\main^{t-1}(1)[a]$ by the definition of the index $a$, these two quantities sum to the total length of $m$ which is $RM$. Therefore, the combined amount their positions increase is (exactly) $RM$ as desired.
\end{proof}

\begin{lemma} \label{lem:ecc-chunk}
    In any chunk $t\in [C]$ where Alice and Bob entered Case~\ref{case:2}:
    
    \begin{enumerate}[label={(\roman*)}]
    \item At the end of the chunk, for any coin $y \in \{ 0, 1 \}^k$, 
    \[
        \pos_\main(y; t) - \pos_\main(y; t-1) 
    \]
    is a lower bound on the amount of corruption the adversary has caused if Alice's input were actually $y$.
    \item It holds that
        \begin{align*}
            &~\posx_\main(1; t) + \posx_\main(2; t) \\
            - &~ \big( \posx_\main(1; t-1) + \posx_\main(2; t-1) \big) \\
            \ge &~\left( \frac12 - 2\epsilon \right)\cdot RM.
        \end{align*}
    \end{enumerate}
\end{lemma}

\begin{proof}
    At the end of the chunk $t$, let us say Bob has received the message $m$ as a length $RM$ concatenation of the messages in each round of the chunk. Bob updated the state of $\main$ by computing the list of values $\cL$ such that $\Delta(\ECC(y),m)\leq \left(\frac12-2\epsilon\right) RM$. For those coins $y\in \cL$, 
    \[
         \pos_\main(y; t) - \pos_\main(y; t-1) = \Delta(\ECC(y),m)
    \]
    and therefore is a lower bound on the corruption the adversary must have caused if Alice sent $\ECC(y)$.
    
    For the other coins $y\in\{0,1\}^k$, 
    \[
         \pos_\main(y; t) - \pos_\main(y; t-1) = \left(\frac12-2\epsilon\right) \cdot RM < \Delta(\ECC(y),m)
    \]
    and therefore is a lower bound on the corruption the adversary must have caused if Alice sent $\ECC(y)$.
    
    To show the second claim, it suffices to show that for any two coins $y$ and $y'$, 
    \begin{align*}
        &~\posx_\main(y; t) + \posx_\main(y'; t) \\
        - &~ (\posx_\main(y; t-1) + \posx_\main(y'; t-1)) \\
        \ge &~\left( \frac12 - 2\epsilon \right)\cdot RM.
    \end{align*}
    
    If this were false, then when Bob updated the values of both coins $y$ and $y'$, both values would have had to be a part of the list $\cL$ of explicit updates. But by the distance condition of Definition~\ref{def:ecc}, it holds that
    \[
        \Delta(\ECC(y),m)+\Delta(\ECC(y'),m) \geq \Delta(\ECC(y),\ECC(y')) \geq \left(\frac12-\epsilon\right)\cdot RM
    \]
    
    so this could not have happened.
\end{proof}

\begin{restatable}{lemma}{lemminigame}\label{lem:coin-minigame-chunk}
In any chunk $t\in[C]$ where Alice and Bob entered Case~\ref{case:3}:
    \begin{enumerate}[label={(\roman*)}]
    \item \label{item:minigame-chunk-lb-corruption}
        At the end of the chunk, for any coin $y \in \{ 0, 1 \}^k$, 
        \[
            \pos_\main(y; t) - \pos_\main(y; t-1)
        \]
        is a lower bound on the amount of corruption the adversary has caused if Alice's input were actually $y$.
    \item \label{item:minigame-chunk-12-2/3} 
        It holds that 
        \begin{align*}
            &~\posx_\main(1; t) + \posx_\main(2; t) \\
            - &~ \big( \posx_\main(1; t-1) + \posx_\main(2; t-1) \big) \\
            \ge &~\left( \frac23 - 4\epsilon \right)\cdot RM.
        \end{align*}
    \item \label{item:minigame-chunk-123-1}
        It holds that 
        \begin{align*}
            &~\posx_\main(1; t) + \posx_\main(2; t) + \posx_\main(3; t) \\
            - &~ \big(\posx_\main(1; t-1) + \posx_\main(2; t-1) + \posx_\main(3; t-1) \big) \\
            \ge &~( 1 - 6\epsilon )\cdot RM.
        \end{align*}
    \end{enumerate}
\end{restatable}

We defer the proof of Lemma~\ref{lem:coin-minigame-chunk} to Section~\ref{sec:minigame-analysis}.

\subsubsection{Analysis of Error Resilience} \label{sec:analysis-error}

Combining Lemmas~\ref{lem:single-bit},~\ref{lem:ecc-chunk}, and~\ref{lem:coin-minigame-chunk}, we obtain the following fact: for any $t\in [C]$ and $y\in \{0,1\}^k$, it holds that $\pos_\main(y; t)$ is a lower bound on the number of bits Eve has corrupted so far. Thus, it now suffices to show that $\posx_\main(2; C) > \left(\frac13-6\epsilon\right)$.

\begin{definition}[$C'$]
    Let $C'\in [C]$ be the chunk right before the first chunk in which Alice and Bob enter Case~\ref{case:1} in the protocol. If this never occurs, let $C':=C$.
\end{definition}

Notice that since $\posx(3;t)$ is non-decreasing over time, if Alice and Bob enter Case~\ref{case:1} in a chunk, they also enter that case in all future chunks. Therefore, for all the chunks $C'+1,\ldots, C$, Alice and Bob enter Case~\ref{case:1}.

\begin{lemma} \label{lem:1-sum}
    At the end of any chunk $t\in [C']$, it holds that 
    \[
        \posx_\main(1;t)+\posx_\main(2;t)+\posx_\main(tRL_\epsilon+1;t)\geq tRM-8\epsilon tRM.
    \]
\end{lemma}

\begin{proof}
    We prove this by induction. Clearly, it holds for $t=0$ because the initial positions are all $0$ and the right-hand side is $0$.
    
    Assuming the statement holds for $t-1$ we prove it for $t$. In particular, we want to show that
    \begin{align*}
        &~\posx_\main(1;t)+\posx_\main(2;t)+\posx_\main(tRL_\epsilon+1;t) \\
        -&~ \big( \posx_\main(1;t-1)+\posx_\main(2;t-1)+\posx_\main((t-1)RL_\epsilon+1;t-1) \big) \\
        \geq &~ RM-8\epsilon RM. \numberthis \label{eqn:1-diff}
    \end{align*} 
    We split the proof into cases based on whether Alice and Bob enter Case~\ref{case:2} or Case~\ref{case:3}.
    
    \begin{caseof}
    \case{Alice and Bob enter Case~\ref{case:2}:}{
    Alice sends $\ECC(x)$ in chunk $t$. Bob computes a set $\cL$ of size $\leq L_\epsilon \le RL_\epsilon$. For every $y\not\in \cL$, it holds that
    \[
        \posx_\main(y;t)-\posx_\main(y;t-1)=\left(\frac12-2\epsilon\right)\cdot RM.
    \]
    Since only $|\cL|$ values of $\posx_\main(y;t) - \posx_\main(y; t-1)$ are an amount $< \left(\frac12-2\epsilon\right)\cdot RM$, for all indices $a\in [2^k-|\cL|]$, it holds that 
    \[
        \posx_\main(a+|\cL|; t) \ge \posx_\main(a;t-1) + \left(\frac12-2\epsilon\right)\cdot RM
    \]
    and so in particular, we have 
    \[
        \posx_\main(tRL_\epsilon+1;t) - \posx_\main((t-1)RL_\epsilon+1;t-1) \geq \left(\frac12-2\epsilon\right)\cdot RM
    \]
    Moreover, by Lemma~\ref{lem:ecc-chunk}, it holds that
        \begin{align*}
            \posx_\main(1; t) + \posx_\main(2; t) - (\posx_\main(1; t-1) + \posx_\main(2; t-1)) \ge \left( \frac12 - 2\epsilon \right)\cdot RM.
        \end{align*}
    Combining these two observations gives Equation~\eqref{eqn:1-diff}.
    }
    \case{Alice and Bob enter Case~\ref{case:3}:}{
    By Lemma~\ref{lem:coin-minigame-chunk}, it holds that
    \[
        \posx_\main(1;t)+\posx_\main(2;t)-\posx_\main(1;t-1)-\posx_\main(2;t-1)\geq \left(\frac23-4\epsilon\right)\cdot RM.
    \]
    
    Moreover, all but at most $RL_\epsilon$ coins that were previously not among $x_\main^{t-1}(1)\ldots x_\main^{t-1}((t-1)RL_\epsilon)$ increase by $\left(\frac13-4\epsilon\right)\cdot M$ in position. Therefore, by the same logic as the previous case it holds that 
    \[
        \posx_\main(tRL_\epsilon+1;t)- \posx_\main((t-1)RL_\epsilon+1;t-1) \geq  \left(\frac13-4\epsilon\right) \cdot RM.
    \]
    }
    Combining these two observations gives Equation~\eqref{eqn:1-diff}. \qedhere
    \end{caseof}
\end{proof}

\begin{lemma} \label{lem:123sum}
    At the end of any chunk $t\in [C']$, it holds that 
    \[
        \posx_\main(1;t)+\posx_\main(2;t) + \posx_\main(3;t) \geq tRM-8\epsilon tRM - 2\epsilon CRM.
    \]
\end{lemma}

\begin{proof}
    We split the proof into cases based on whether Alice and Bob entered Case~\ref{case:2} or Case~\ref{case:3} during chunk $t$.
    
    \begin{caseof}
    \case{Alice and Bob entered Case~\ref{case:2}:}{
        In this case, we have
        \begin{align*}
            &~ \posx_\main(1;t)+\posx_\main(2;t)+\posx_\main(3;t) \\
            \geq &~ \posx_\main(1;t-1)+\posx_\main(2;t-1)+\pos_\main(3;t-1) \\
            \geq &~ \posx_\main(1;t-1)+\posx_\main(2;t-1)+\posx_\main((t-1)RL_\epsilon+1;t-1)-\epsilon CRM.
        \end{align*}
        Using Lemma~\ref{lem:1-sum} on chunk $t-1$, we further get
        \begin{align*}
            &~ \posx_\main(1;t-1)+\posx_\main(2;t-1)+\posx_\main((t-1)RL_\epsilon+1;t-1)-\epsilon CRM \\
            \geq &~ (t-1)RM -8\epsilon (t-1)RM -\epsilon CRM \\
            \geq &~ tRM - 8\epsilon tRM - \epsilon CRM - RM \\
            \geq &~ tRM -8\epsilon tRM -2\epsilon CRM,
        \end{align*}
        which gives the desired result.
    }
    
    \case{Alice and Bob entered Case~\ref{case:3}:}{
        In this case, Alice and Bob play the coin minigame. We prove the claim by induction. Clearly, the claim holds for $t=0$. Assuming the claim for $t-1$, we have
        \begin{align*}
            &~\posx_\main(1;t-1)+\posx_\main(2;t-1)+\posx_\main(3;t-1) \\
            \geq &~ (t-1)RM-8\epsilon (t-1)RM - 2\epsilon CRM.
        \end{align*}
        By Lemma~\ref{lem:1-sum}, it holds that
        \begin{align*}
            &~\posx_\main(1; t) + \posx_\main(2; t) + \posx_\main(3; t) \\
            - &~(\posx_\main(1; t-1) + \posx_\main(2; t-1) + \posx_\main(3; t-1)) \\
            \ge &~(1 - 6\epsilon )\cdot RM.
        \end{align*}
        Then, the claim follows by adding the two equations. \qedhere
    }
    \end{caseof}
\end{proof}

\begin{proof}[Proof of Theorem~\ref{thm:error-feedback-ecc}]
    We aspire to show that 
    \[
        \posx_\main(2;C) > \left(\frac13-6\epsilon\right)\cdot CRM.
    \]
    Assume for the sake of contradiction this is not true at the end of the protocol, so in particular at least two coins remain on the board. At the end of chunk $C'$, it holds by Lemma~\ref{lem:123sum} that
    \[
        \posx_\main(1;C')+\posx_\main(2;C') + \posx_\main(3;C') \geq C'RM-8\epsilon C'RM - 2\epsilon CRM.
    \]
    Moreover, we have that in the chunk $C'$ because Alice and Bob did not enter Case~\ref{case:1}, it held that 
    \begin{align*}
        &~\posx_\main(3;C'-1)\leq \frac13 \cdot CRM \\
        \implies &~ \posx_\main(3;C')\leq \frac13 \cdot CRM +RM = \left(\frac13+\epsilon\right)\cdot CRM
    \end{align*}
    at the start of chunk $C'$.
    Therefore, at the end of chunk $C'$, it holds that
    \begin{align*}
        &~ \posx_\main(1;C')+\posx_\main(2;C') + \left(\frac13+\epsilon\right)\cdot CRM \geq C'RM-8\epsilon C'RM - 2\epsilon CRM \\
        \implies &~ \posx_\main(1;C')+\posx_\main(2;C') \geq C'RM - \frac13 \cdot CRM-8\epsilon C'RM - 3\epsilon CRM
    \end{align*}
    In the chunks $t\in [C'+1, C]$, it holds that 
    \begin{align*}
        &~ \posx_\main(1;t)+\posx_\main(2;t)-(\posx_\main(1;t-1)+\posx_\main(2;t-1)) \geq RM \\
        \implies &~ \posx_\main(1;C)+\posx_\main(2;C)-(\posx_\main(1;C')+\posx_\main(2;C')) \geq CRM-C'RM.
    \end{align*}
    Note that this holds even if $C'=C$ and Alice and Bob never entered Case~\ref{case:1}. Adding these equations gives 
    \begin{align*}
        & \posx_\main(1;C)+\posx_\main(2;C) \geq \frac23 \cdot CRM-8\epsilon C'RM - 3\epsilon CRM \geq \left(\frac23-11\epsilon\right)\cdot CRM \\
        \implies & 2\posx_\main(2;C') \geq \left(\frac23-11\epsilon\right)\cdot CRM \\
        \implies & \posx_\main(2;C') > \left(\frac13-6\epsilon\right)\cdot CRM \ .
    \end{align*}
    Therefore, at the end of the protocol, there is only one possible input that Alice could have had, namely $x_\main^C(1)$, that required at most $\left(\frac13-6\epsilon\right)\cdot CRM$ corruptions from the adversary given the messages Bob received. Therefore, Bob's output $x_\main^C(1)$ must indeed have been Alice's true input.
\end{proof}
    
\subsubsection{Coin Minigame Analysis}\label{sec:minigame-analysis}

In this section, we prove Lemma~\ref{lem:coin-minigame-chunk}. 

We begin with the following observation relating the positions of the coins in the minigame to the amount of corruption caused by the adversary if Alice's input were a given coin in the main game.

\begin{lemma} \label{lem:minigame-chunk-lb}
    At the end of a chunk $t\in [C]$ in which the minigame is played, then for any $y \in \{ 0, 1 \}^k$, if Alice's input were actually $y$, the adversary must have caused at least $\pos_\mini(\partdec(s, y); R)$ corruption to Alice's communication. 
\end{lemma}

\begin{proof}
    Note that after round $\tau$ in chunk $t$, the position of coin $i \in [(t-1)RL_\epsilon]$ has increased by 
    \[
        v^\tau(i)=
        \begin{cases}
            \min_{z\in\Lambda_\tau}\Delta(\ECC[s,f^\tau](z),m)  & f^\tau(i)=1 ~\text{and}~ \Lambda_\tau \not= \emptyset\\
            \left( \frac23 - \epsilon \right) \cdot M - \min \{ d_2^\tau, d_3^\tau \} & f^\tau(i)=1 ~\text{and}~ \Lambda_\tau = \emptyset\\
            d_{f^\tau(i)}^\tau & f^\tau(i)\neq 1,
        \end{cases}
    \]
    where $d_2^\tau := \Delta(\ECC[s, f^\tau](x_\main^{t-1}(2)), m)$ and $d_3^\tau := \Delta(\ECC[s, f^\tau](x_\main^{t-1}(3)), m)$. 
    
    In the case $f^\tau(\partdec(s, y)) = 1$ and $\Lambda_\tau \not= \emptyset$, we note that $\min_{z \in \Lambda_\tau} \Delta(\ECC[s,f^\tau](z), m) \le \Delta(\ECC[s,f^\tau](y), m)$ since $\Lambda_\tau$ consists of all values of $z$ for which $\Delta(\ECC[s,f^\tau](z), m) < \left( \frac13 - 2\epsilon \right) \cdot M$. Thus, the update value $v^\tau(i)$ is a lower bound on the actual amount of corruption the adversary caused. 
    
    In the case $f^\tau(\partdec(s, y)) = 1$ and $\Lambda_\tau = \emptyset$, we have that 
    \[
        \Delta(\ECC[s,f^\tau](y), m) + \Delta(\ECC[s,f^\tau](x_\main^{t-1}(a)), m) \ge \left( \frac23 - \epsilon \right) \cdot M
    \]
    for any $a \in [2 : (t-1)RL_\epsilon]$. In particular, 
    \[
        \Delta(\ECC[s,f^\tau](y), m) \ge \left( \frac23 - \epsilon \right) \cdot M - \min \{ d_2^\tau, d_3^\tau \}. 
    \]
    
    Finally, if $f^\tau(\partdec(s, y)) \not= 1$, letting $i = \partdec(s, y)$, note that $\ECC[s,f^\tau](y) = \ECC[s,f^\tau](x_\main^{t-1}(i)) = \ECC[s,f^\tau](x_\main^{t-1}(f^\tau(i)))$ by definition of $\ECC[s,f^\tau]$ (see Definition~\ref{def:ecc3}). Thus, $v^\tau(i)$ is precisely the amount of corruption the adversary must've caused.
\end{proof}

With this view in mind, we proceed to analyze the coin minigame. We will return afterward to the proof of Lemma~\ref{lem:coin-minigame-chunk}.

Our first goal is to prove that the sum of positions of any two coins at the end of the minigame in chunk $t$ is at least $\approx \frac23 RM$. To do this, we will consider the following new coin game, defined on pairs of coins in the minigame. We call this new coin game $\pair$. 
\begin{itemize}
    \item 
        The coins are all pairs $(i,j) \in \binom{S_\mini}{2} =: S_\pair$, where $S_\mini = [(t-1)RL_\epsilon]$. We consider pairs of coins without ordering, so $(i, j)$ and $(j, i)$ refer to the same pair of coins.
    \item 
        After any round $\tau$, the position of coin $(i,j)$ is the sum of positions of the coins $i$ and $j$ in the minigame, i.e.,
        \[
            \pos_\pair(i, j; \tau) = \pos_\mini(i; \tau) + \pos_\mini(j; \tau). 
        \]
    \item 
        It follows that in any round $\tau$, coin $(i, j)$ in the $\pair$ game is updated by $v^\tau(i, j) := v^\tau(i) + v^\tau(j)$, where $v^\tau(i)$ is the position update of coin $i$ in the minigame.
    \item 
        We will also refer to coins in $\pair$ by the relative ordering of the positions of the two coins in $\mini$. That is, we will denote by $x_\pair^\tau(a,b)$ the $\pair$ coin $(x_\mini^\tau(a), x_\mini^\tau(b))$, for $(a,b) \in \binom{|S_\mini|}{2}$. We also let $\posx_\pair(a,b; \tau)$ denote $\pos_\pair(x_\pair^\tau(a,b); \tau) = \posx_\mini(a; \tau) + \posx_\mini(b; \tau)$. For convenience, we will also let $v^\tau[a,b]$ denote $v^\tau(x_\pair^{\tau-1}(a,b))$. 
\end{itemize}

For $a,b \in \binom{|S_\mini|}{2}$, we impose the following ordering: $(a, b) \le (a', b')$ (where $a < b$ and $a' < b'$) if and only if $a \le a'$ and $b \le b'$, with strict inequality if either $a < a'$ or $b < b'$. We remark that we are considering pairs of indices $(a,b)$ without ordering, so $(a, b)$ is the same as $(b, a)$.

Our goal is then to show that at any round $\tau$, all coins in $\pair$ are at position $\gtrsim \frac23 \tau$.

We begin by showing that for any $i, j$ such that $f^\tau(i) \not= f^\tau(j)$, the total update of coins $i$ and $j$ in the $\mini$ game is at least $\approx \frac23 M$. 

\begin{lemma} \label{lem:minigame-ij-2/3}
    In any chunk $t \in [C]$ in which the $\mini$ game is played, and for any round $\tau \in [R]$, it holds that for any $i, j \in [(t-1)RL_\epsilon]$ such that $f^\tau(i) \not= f^\tau(j)$, 
    \[
        v^\tau(i) + v^\tau(j) \ge \left( \frac23 - \epsilon \right) \cdot M.
    \]
\end{lemma}

\begin{proof}
    There are a few cases. The first case is that $f^\tau(i) = 2$ and $f^\tau(j) = 3$ (or vice versa). Then 
    \begin{align*}
        v^\tau(i) + v^\tau(j) 
        &= d^\tau_2 + d^\tau_3 \\
        &\ge \Delta(\ECC[s,f^\tau](x_\main^{t-1}(2)), m) + \Delta(\ECC[s,f^\tau](x_\main^{t-1}(3)), m) \\
        &\ge \left( \frac23 - \epsilon \right) \cdot M
    \end{align*}
    by property~\ref{item:ecc-2/3} of Lemma~\ref{def:ecc3}.
    
    The second case is that $f^\tau(i) = 1$ and $f^\tau(j) \in \{ 2, 3 \}$ (or vice versa). Then, if $\Lambda_\tau \not= \emptyset$, it holds that 
    \begin{align*}
        v^\tau(i) + v^\tau(j) 
        &= \min_{z \in \Lambda_\tau} \Delta(\ECC[s,f^\tau](z), m) + \Delta(\ECC[s,f^\tau](x_\main^{t-1}(f^\tau(j))), m) \\
        &\ge \left( \frac23 - \epsilon \right) \cdot M
    \end{align*}
    once again by property~\ref{item:ecc-2/3} of Lemma~\ref{def:ecc3}.
    On the other hand, if $\Lambda_\tau = \emptyset$, then 
    \begin{align*}
        v^\tau(i) + v^\tau(j) 
        &= \left( \frac23 - \epsilon \right) \cdot M - \min \{ d_2^\tau, d_3^\tau \} + d_{f^\tau(i)}^\tau \\
        &\ge \left( \frac23 - \epsilon \right) \cdot M,
    \end{align*}
    where here we note that $f^\tau(i) \in \{ 2, 3 \}$ so that $d^\tau_{f^\tau(i)} \in \{ d_2^\tau, d_3^\tau \}$.
\end{proof}

In the following lemma, we will show a property of any update pattern $v^\tau$ in the $\pair$ game, that the coins in the $\pair$ game can be paired up such that the total update in the two coins is at least $\approx \frac43 M$. 

\begin{lemma}
    Given any update function $v^\tau : S_\pair \rightarrow [0,2M]$, there exists a bijection\footnote{We will assume that $|S_\mini| = (t-1)RL_\epsilon$ is even. This can be done for instance by simply choosing $R$ to be even.} of indices $\binom{|S_\mini|}{2} \ni (a,b) \leftrightarrow (a',b') \in \binom{|S_\mini|}{2}$ such for any bijected values $(a,b) \leftrightarrow (a',b')$,
    \[
        v^\tau[a,b] + v^\tau[a',b'] \ge \left( \frac43 - 2\epsilon \right) \cdot M.
    \]
    Furthermore, for any $(a,b) \in \binom{|S_\mini|}{2}$ with $v^\tau[a,b] < \frac23-\epsilon$, the associated $(a',b')$ satisfies $(a',b') < (a,b)$. 
\end{lemma}

\begin{proof}
    It suffices to find the bijection for indices $(a,b)$ where $v^\tau[a,b] < \frac23 - \epsilon$. Afterwards, for all un-bijected indices, one can pair them up arbitrarily: then $v^\tau[a,b] + v^[a',b'] \ge \left( \frac23 - \epsilon \right) \cdot M + \left( \frac23 - \epsilon \right) \cdot M$. We thus concentrate on $(a,b)$ for which $v^\tau[a,b] < \frac23 - \epsilon$. 
    
    In this case, we write $x_\pair^{\tau-1}(a,b) = (i, j)$. By Lemma~\ref{lem:minigame-ij-2/3}, it must hold that $f^\tau(i) = f^\tau(j)$, where $f^\tau : [(t-1)RL_\epsilon] \rightarrow \{ 1, 2, 3 \}$ is the function Bob sent in round $\tau$. Recall also that $f^\tau(x_\mini^{\tau-1}(1)) = 1$, $f^\tau(x_\mini^{\tau-1}(2,4,6,\dots)) = 2$, and $f^\tau(x_\mini^{\tau-1}(3,5,7,\dots)) = 3$. Thus, since $a \not= b$, we have that either $a$ and $b$ are both even, or $a$ and $b$ are both odd and greater than $1$.
    
    In either case, consider $(a',b')$ defined as $a' = a-1$ and $b'=b-1$. Note in particular that $f(x_\mini^{\tau-1}(a)) \not= f(x_\mini^{\tau-1}(a'))$, and $f(x_\mini^{\tau-1}(b)) \not= f(x_\mini^{\tau-1}(b'))$. Then it holds that 
    \begin{align*}
        v^\tau[a,b] + v^\tau[a',b'] 
        =&~ v^\tau(x_\mini^{\tau-1}(a)) + v^\tau(x_\mini^{\tau-1}(b)) + v^\tau(x_\mini^{\tau-1}(a')) + v^\tau(x_\mini^{\tau-1}(b')) \\
        \ge&~ \left( \frac23 - \epsilon \right) \cdot M + \left( \frac23 - \epsilon \right) \cdot M \\
        =&~ \left( \frac43 - 2\epsilon \right) \cdot M
    \end{align*}
    by appealing to Lemma~\ref{lem:minigame-ij-2/3}. 
    
    Finally, notice that the map $(a, b) \mapsto (a',b')$ as defined above for $v^\tau[a,b] < \frac23 - \epsilon$ is injective. (Also since $v^\tau[a,b] + v^\tau[a',b'] \ge \left( \frac43 - 2\epsilon \right) \cdot M$, it holds that $v^\tau[a',b'] > \left( \frac23 - \epsilon \right) \cdot M$, so the pair $(a',b')$ doesn't map to a different pair.) Thus, we have found our bijection.
\end{proof}




\begin{lemma} \label{lem:pair12-2/3}
    After any number of $\tau \le R$ rounds in chunk $t$, it holds that 
    \[
        \posx_\pair(1,2; \tau) 
        = \posx_\mini(1; \tau) + \posx_\mini(2; \tau)
        \ge \left( \frac23 - \epsilon \right) \cdot \tau M - 3\epsilon R M. 
    \]
\end{lemma}

\begin{proof}
    Consider the potential function 
    \[
        \phi(\tau) := \sum_{(i,j) \in \binom{S_\mini}{2}} w(i,j; \tau),
    \]
    where
    \[
        w(i,j; \tau) := (1 + \epsilon)^{\frac3{4M} \cdot d(i,j;\tau)}
    \]
    and
    \[
        d(i,j;\tau) := \max \left\{ 0, \left( \frac23 - \epsilon \right) \cdot \tau M - \pos_\pair(i,j;\tau) \right\}.
    \]
    Since $\pos_\pair(i, j; 0) = 0$ for all $i, j \in \binom{S_\mini}{2}$, all values of $d(i,j; \tau)$ are $0$ at time $\tau = 0$. Then $\phi(0) = \binom{|S_\mini|}{2}$.
    
    Consider any bijected pair of values $(a, b)$ and $(a', b')$ for the update function $v^\tau$, and let $i_a = x_\mini^{\tau-1}(a)$, $i_b = x_\mini^{\tau-1}(b)$, $i_{a'} = x_\mini^{\tau-1}(a')$, and $i_{b'} = x_\mini^{\tau-1}(b')$. 
    
    \begin{claim} \label{claim:w-potential}
        It holds that
        \[
            w(i_a, i_b; \tau) + w(i_{a'}, i_{b'}; \tau) \le e^{\epsilon^2/2} \cdot \big( w(i_a, i_b; \tau-1) + w(i_{a'}, i_{b'}; \tau-1) \big) + \epsilon.
        \]
    \end{claim}
    
    \begin{proof}
        To see why this claim holds, we will consider two cases based on the values of $v^\tau[a,b]$ and $v^\tau[a',b']$.
        
        \begin{caseof}
        \case{Both $v^\tau[a,b]$ and $v^\tau[a',b']$ are $\ge \left( \frac23 - \epsilon \right) \cdot M$.}{
            It follows that  $d(i_a, i_b; \tau) \le d(i_a, i_b; \tau-1)$ and $d(i_{a'}, i_{b'}; \tau) \le d(i_{a'}, i_{b'}; \tau - 1)$, so 
            \begin{align*}
                w(i_a, i_b; \tau) + w(i_{a'}, i_{b'}; \tau) 
                &\le w(i_a, i_b; \tau-1) + w(i_{a'}, i_{b'}; \tau-1) \\
                &\le e^{\epsilon^2/2} \cdot \big( w(i_a, i_b; \tau-1) + w(i_{a'}, i_{b'}; \tau-1) \big) + \epsilon.
            \end{align*}
        }
        \case{$v^\tau[a,b] < \left( \frac23 - \epsilon \right) \cdot M$.}{
            Write $v^\tau[a,b] = \left( \frac23 - \epsilon \right) \cdot M - \delta$, where $\delta \le \left( \frac23 - \epsilon \right) \cdot M$. It holds that 
            \[
                d(i_a, i_b; \tau) \le d(i_a, i_b; \tau-1) + \delta.
            \]
            Furthermore, $v^\tau[a',b'] \ge \left( \frac23 - \epsilon \right) \cdot M + \delta$. 
            We consider two subcases.
            
            \begin{subcaseof}
                \subcase{$d(i_{a'},i_{b'}; \tau-1) \ge \delta$}{
                    In this case, we have $d(i_{a'}, i_{b'}; \tau) \le d(i_{a'},i_{b'}; \tau-1) - \delta$. Then 
                    \begin{align*}
                        w(i_a,i_b; \tau) + w(i_{a'},i_{b'}; \tau) 
                        &\le (1+\epsilon)^{\frac3{4M} \cdot \delta} \cdot w(i_a, i_b; \tau-1) + (1+\epsilon)^{-\frac3{4M} \cdot \delta} \cdot w(i_{a'}, i_{b'}; \tau-1) \\
                        &\le \frac{(1+\epsilon)^{\frac3{4M} \cdot \delta} + (1+\epsilon)^{-\frac3{4M} \cdot \delta}}{2} \cdot \big( w(i_a,i_b; \tau-1) + w(i_{a'},i_{b'}; \tau-1) \big) \\
                        &\le \frac{(1+\epsilon) + (1+\epsilon)^{-1}}{2} \cdot \big( w(i_a,i_b; \tau-1) + w(i_{a'},i_{b'}; \tau-1) \big) \\
                        &\le e^{\epsilon^2/2} \cdot \big( w(i_a,i_b; \tau-1) + w(i_{a'},i_{b'}; \tau-1) \big) \\
                        &\le e^{\epsilon^2/2} \cdot \big( w(i_a,i_b; \tau-1) + w(i_{a'},i_{b'}; \tau-1) \big) + \epsilon,
                    \end{align*}
                    where we use in the second inequality convexity along with the fact that $(a',b') < (a,b)$ and therefore $d(i_a,i_b;\tau-1) \le d(i_{a'},i_{b'}'; \tau-1) \implies w(i_a,i_b;\tau-1) \le w(i_{a'},i_{b'};\tau-1)$; in the third inequality we use that $\delta < \frac23 M \implies \frac3{4M} \cdot \delta < \frac12 < 1$ and convexity of the function $\frac{(1+\epsilon)^z + (1+\epsilon)^{-z}}{2}$; and in the fourth inequality we use that $\frac{(1+\epsilon) + (1+\epsilon)^{-1}}{2} \le 1 + \frac{\epsilon^2}{2} \le e^{\epsilon^2/2}$. 
                }
                \subcase{$d(i_{a'},i_{b'};\tau-1) < \delta$.}{
                    In this case, we have that $d(i_{a'},i_{b'};\tau) = 0$. Also since $d(i_{a'},i_{b'};\tau-1) \ge d(i_a,i_b; \tau-1)$, we know that $d(i_a, i_b;\tau) < 2\delta$. This implies that 
                    \begin{align*}
                        w(i_a,i_b; \tau) + w(i_{a'},i_{b'}; \tau) 
                        &\le (1+\epsilon)^0 + (1+\epsilon)^{\frac3{4M} \cdot 2\delta} \\
                        &\le 2+\epsilon \\
                        &\le \big( w(i_a,i_b; \tau-1) + w(i_{a'},i_{b'}; \tau-1) \big) + \epsilon \\
                        &\le e^{\epsilon^2/2} \cdot \big( w(i_a,i_b; \tau-1) + w(i_{a'},i_{b'}; \tau-1) \big) + \epsilon. \qedhere
                    \end{align*}
                }
            \end{subcaseof}
        }
        \end{caseof}
    \end{proof}
    
    We return now to the proof of Lemma~\ref{lem:pair12-2/3}. We can add up $w(i_a, i_b; \tau) + w(i_{a'}, i_{b'}; \tau)$ as given in Claim~\ref{claim:w-potential} for all bijected pairs $(a, b) \leftrightarrow (a', b')$ to obtain $\phi(\tau)$, giving us that 
    \[
        \phi(\tau) \le e^{\epsilon^2/2} \cdot \phi(\tau-1) + \epsilon \cdot \frac{\binom{|S_\mini|}{2}}{2}. 
    \]
    Solving the recurrence with $\phi(0) = \binom{|S_\mini|}{2}$, we get
    \begin{align*}
        \phi(\tau) 
        &\le e^{\tau \epsilon^2/2} \cdot \binom{|S_\mini|}{2} + \frac{e^{\tau \epsilon^2/2}-1}{e^{\epsilon^2/2}-1} \cdot \epsilon \cdot \frac{\binom{|S_\mini|}{2}}{2} \\
        &\le e^{\tau\epsilon^2/2} \cdot \binom{|S_\mini|}{2} \cdot \left( 1 + \frac{\epsilon}{2(e^{\epsilon^2/2}-1)} \right) \\
        &\le e^{R\epsilon^2/2} \cdot \binom{|S_\mini|}{2} \cdot \left( 1 + \frac{1}{\epsilon} \right).
    \end{align*}
    Finally, consider the pair $(i, j)$ where $i = x_\mini^\tau(1)$ and $j = x_\mini^\tau(2)$. We have that
    \begin{align*}
        w(i,j; \tau) 
        = (1+\epsilon)^{\frac3{4M} \cdot d(i,j;\tau)} 
        &\le \phi(\tau) \\
        &\le e^{R\epsilon^2/2} \cdot \binom{|S_\mini|}{2} \cdot \left( 1 + \frac{1}{\epsilon} \right) \\
        &\le e^{R\epsilon^2/2} \cdot (CRL_\epsilon)^2 \cdot \left( 1 + \frac{1}{\epsilon} \right),
    \end{align*}
    which gives 
    \begin{align*}
        \left( \frac23 - \epsilon \right) \cdot \tau M - \posx_\pair(1,2; \tau) 
        &\le d(i,j;\tau) \\
        &\le \frac{4M}{3 \ln(1+\epsilon)} \cdot \left( R\epsilon^2/2 + 2 \ln (CRL_\epsilon) + \frac{1}{\epsilon} \right) \\
        &\le \frac{4M}{\epsilon} \cdot \left( R\epsilon^2/2 + 2 \ln (CRL_\epsilon) + \frac{1}{\epsilon} \right) \\
        &= 2\epsilon RM + M \cdot \frac{4}{\epsilon} \cdot \left( 2 \ln (CRL_\epsilon) + \frac{1}{\epsilon} \right) \\
        &\le 3\epsilon RM,
    \end{align*}
    where we use that $\frac\epsilon3 \le \ln \left( 1 + \frac{1}{\epsilon} \right) \le \frac{1}{\epsilon}$, and where the last inequality follows since 
    \begin{align*}
        R &= \frac{100 \ln L_\epsilon}{\epsilon^3} \\
        &\ge \frac{8\ln(1/\epsilon)}{\epsilon^2} + \frac{8 \cdot \left( \ln 100 + \ln L_\epsilon \right)}{\epsilon^2} + \frac{8 \ln L_\epsilon}{\epsilon^2} + \frac{4}{\epsilon^3} \\
        &\ge \frac{8\ln C}{\epsilon^2} + \frac{8 \ln R}{\epsilon^2} + \frac{8 \ln L_\epsilon}{\epsilon^2} + \frac{4}{\epsilon^3} \\
        &\ge \frac{4}{\epsilon^2} \cdot \left( 2 \ln (CRL_\epsilon) + \frac{1}{\epsilon} \right). \qedhere
    \end{align*}
\end{proof}

\begin{corollary} \label{cor:triple123-1}
    After the $R$ rounds of chunk $t$, it holds that 
    \[
        \posx_\mini(1; R) + \posx_\mini(2; R) + \posx_\mini(3; R) \ge \left( 1 - 6\epsilon \right) \cdot RM.
    \]
    In particular, for any pairwise distinct $i,j,k \in [(t-1)RL_\epsilon]$, it holds that
    \[
        \pos_\mini(i; R) + \pos_\mini(j; R) + \pos_\mini(k; R) \ge \left( 1 - 6\epsilon \right) \cdot RM.
    \]
\end{corollary}

\begin{proof}
    By Lemma~\ref{lem:pair12-2/3}, 
    \[
        \posx_\mini(1; R) + \posx_\mini(2; R) \ge \left( \frac23 - 4\epsilon \right) \cdot RM.
    \]
    Since $\posx_\mini(3; R) \ge \posx_\mini(2; R) \ge \posx_\mini(1; R)$, it follows that 
    \[
        \posx_\mini(3; R) \ge \left( \frac13 - 2\epsilon \right) \cdot RM,
    \]
    from which it follows that
    \[
        \posx_\mini(1; R) + \posx_\mini(2; R) + \posx_\mini(3; R) \ge \left( 1 - 6\epsilon \right) \cdot RM \ . \qedhere
    \]
\end{proof}

We finally return to the proof of Lemma~\ref{lem:coin-minigame-chunk}, restated below for convenience. 

\lemminigame*

\begin{proof}[Proof of~\ref{item:minigame-chunk-lb-corruption}]
    By Lemma~\ref{lem:minigame-chunk-lb}, the amount of corruption done by the adversary throughout the chunk if Alice's input were actually $y$ is at least $\pos_\mini(\partdec(s,y); R)$. So for all $y \in \cL$, $\pos_\mini(\partdec(s,y); R) = \pos_\main(y; t) - \pos_\main(y; t-1)$ is a lower bound on the amount of corruption the adversary has caused.
    
    Now, let us consider the other values of $y$. At the end of the $R$ rounds, let $\hat{i}$ be the value of $i \in [(t-1)RL_\epsilon]$ that minimizes $\pos_\mini(i; R)$. If $\partdec(s,y) \not= \hat{i}$, then by Lemma~\ref{lem:pair12-2/3} it holds that 
    \begin{align*}
        \pos_\mini(\partdec(s,y); R) 
        &\ge \frac{\left( \frac23 - \epsilon \right) \cdot R M - 3\epsilon RM}{2} \\
        &\ge \left( \frac13 - 2\epsilon \right) \cdot RM \\
        &\ge \left( \frac13 - 4\epsilon \right) \cdot RM,
    \end{align*}
    so $\pos_\main(y;t) - \pos_\main(y;t-1)$ lower bounds the corruption done by the adversary.
    
    Finally, if $\partdec(s,y) = \hat{i}$, let $R'$ be the last round after which $\pos_\mini(\hat{i}; \tau)$ was smallest out of all the $\mini$ game coin positions for the rest of the chunk (so immediately after round $R'$, $\pos_\mini(\hat{i}; R')$ was not yet the smallest). Then by Lemma~\ref{lem:pair12-2/3}, it holds that 
    \[
        \pos_\mini(\hat{i}; R') 
        \ge \frac{\left( \frac23 - \epsilon \right) \cdot R'M - 3\epsilon RM}{2}.
    \]
    In rounds $R'+2$ and after, Bob always sends a function $f^\tau : [(t-1)RL_\epsilon] \rightarrow \{ 1, 2, 3 \}$ satisfying $f^\tau(\hat{i}) = 1$. Since $y \not\in \bigcup_\tau \Lambda_\tau$, it holds that for any message $m$ received in round $\tau \ge R'+2$, $\Delta(\ECC[s,f^\tau](y), m) \ge \left( \frac13 - 2\epsilon \right) \cdot M$. Altogether, the total corruption on $y$ is lower bounded by
    \begin{align*}
        \frac{\left( \frac23 - \epsilon \right) \cdot R'M - 3\epsilon RM}{2} &+ \left( \frac13 - 2\epsilon \right) \cdot M \cdot (R - R' - 1) \\
        &= \left( \frac13 - \frac72 \epsilon \right) \cdot RM + \frac32 \epsilon R'M - \left( \frac13 - 2\epsilon \right) \cdot M \\
        &\ge \left( \frac13 - \frac72\epsilon \right) \cdot RM - \left( \frac13 - 2\epsilon \right) \cdot \epsilon RM \\
        &\ge \left( \frac13 - 4\epsilon \right) \cdot RM \ . \qedhere
    \end{align*}
\end{proof}

\begin{proof}[Proof of~\ref{item:minigame-chunk-12-2/3} and~\ref{item:minigame-chunk-123-1}]
    At the end of a chunk $t$ in which the minigame is played, the largest possible position of a coin in the minigame is $RM$, so the largest possible update to any coin $y \in \{ 0, 1 \}^k$ is $\pos_\main(y; t) = \pos_\main(y; t-1) + RM$. However, recall that if the minigame were played in chunk $t$, it must mean that 
    \begin{align*}
        \posx_\main(1; t-1), \posx_\main(2; t-1), \posx_\main(3; t-1) 
        &\le \posx_\main((t-1)RL_\epsilon+1; t-1) - \epsilon RCM \\
        &= \posx_\main((t-1)RL_\epsilon+1; t-1) - RM,
    \end{align*}
    using that $C = \frac1\epsilon$. In particular, this implies that 
    \[
        \pos_\main(x_\main^{t-1}(1); t), \pos_\main(x_\main^{t-1}(2); t), \pos_\main(x_\main^{t-1}(3); t)
        \le \pos_\main(x_\main^{t-1}(i); t)
    \]
    for all $i > (t-1)RL_\epsilon$. 
    Thus, if we want to find the positions of the first two or three coins after chunk $t$, it suffices to consider only the resulting positions of the first $(t-1)RL_\epsilon$ coins after chunk $t-1$ at time $t$.
    
    Now, for~\ref{item:minigame-chunk-12-2/3},
    \begin{align*}
        \posx_\main(1; t) + \posx_\main(2; t) 
        =&~ \left( \begin{aligned}
            \pos_\main(x_\main^t(1); t) + \pos_\main(x_\main^t(2); t) \\
            - \pos_\main(x_\main^t(1); t-1) - \pos_\main(x_\main^t(2); t-1)
        \end{aligned} \right) \\
        &+ \pos_\main(x_\main^t(1); t-1) + \pos_\main(x_\main^t(2); t-1) \\
        \ge&~ \left( \frac23 - 4\epsilon \right) \cdot RM + \pos_\main(x_\main^t(1); t-1) + \pos_\main(x_\main^t(2); t-1) \\
        \ge&~ \left( \frac23 - 4\epsilon \right) \cdot RM + \posx_\main(1; t-1) + \posx_\main(2; t-1),
    \end{align*}
    where the first inequality holds due to 
    \begin{align*}
        \left( \begin{aligned}
            \pos_\main(x_\main^t(1); t) + \pos_\main(x_\main^t(2); t) \\
            - \pos_\main(x_\main^t(1); t-1) - \pos_\main(x_\main^t(2); t-1)
        \end{aligned} \right)
        &= \pos_\mini(\partdec(s, x_\main^t(1))) + \pos_\mini(\partdec(s, x_\main^t(2))) \\
        &\ge \left( \frac23 - 4\epsilon \right) \cdot RM.
    \end{align*} by Lemma~\ref{lem:pair12-2/3}.
    
    As for~\ref{item:minigame-chunk-123-1}, we similarly write
    \begin{align*}
        \posx_\main(1; t) &+ \posx_\main(2; t) + \posx_\main(3; t) \\ 
        =&~ \left( \begin{aligned}
            \pos_\main(x_\main^t(1); t) + \pos_\main(x_\main^t(2); t) + \pos_\main(x_\main^t(3); t) \\
            - \pos_\main(x_\main^t(1); t-1) - \pos_\main(x_\main^t(2); t-1) - \pos_\main(x_\main^t(3); t-1)
        \end{aligned}\right) \\
        &+ \pos_\main(x_\main^t(1); t-1) +  \pos_\main(x_\main^t(2); t-1) + \pos_\main(x_\main^t(3); t-1) \\
        \ge&~ \left( 1 - 6\epsilon \right) \cdot RM + \pos_\main(x_\main^t(1); t-1) +  \pos_\main(x_\main^t(2); t-1) + \pos_\main(x_\main^t(3); t-1) \\
        \ge&~ \left( 1 - 6\epsilon \right) \cdot RM + \posx_\main(1; t-1) + \posx_\main(2; t-1) + \posx_\main(3; t-1), 
    \end{align*}
    using Corollary~\ref{cor:triple123-1}. 
\end{proof}

\section{Lower Bounds}\label{sec:lower-bounds}
\subsection{Lower Bound on the Number of Bits of Feedback} \label{sec:lb-bits}

Suppose that Alice has input $x \in \{ 0, 1 \}^k$, which she is trying to convey to Bob with feedback from Bob. For error-correcting codes without feedback, it is well known the maximum error (resp. erasure) resilience the parties can achieve for arbitrarily large $k$ is $\frac14$ (resp. $\frac12$) by the Plotkin bound. Here we show that in order surpass these thresholds, Bob must send at least $\Omega(\log{k})$ bits of feedback. With only $o(\log{k})$ bits of feedback, the parties cannot achieve higher error resilience than they could without feedback.

\begin{theorem}
    For any $\delta>0$, any feedback ECC family $\{\code\}_{k \in \bbN}$ that achieves error resilience $\frac14 +\delta$ requires $\Omega_\delta(\log{k})$ bits of feedback. Similarly, any $\code$ that achieves erasure resilience $\frac12 +\delta$ requires $\Omega_\delta(\log{k})$ bits of feedback.
\end{theorem}

\begin{proof}
    Throughout the proof, fix $k$ and let $\code$ denote $\code_k$. Let $G$ be the complete graph whose vertices are all of Alice's possible messages $\in \{ 0, 1 \}^k$. We denote the vertices $v_1\ldots v_K$ where $K=2^k$. On each input $v_1\ldots v_K$, let $\code(v_i;f)$ denote the sequence of bits Alice sends throughout the protocol, where $f\in \{0,1\}^{\zeta(\code)}$ denotes the feedback bits she gets from Bob.
    
    In the graph $G$, we label the edge $(v_i,v_j)$ with $f\in \{0,1\}^{\zeta(\code)}$ if $f$ is the lexicographically first string for which
    \[
        \Delta \left( \code(v_i;f), \code(v_j;f) \right) \geq \left( \frac12+\delta \right)|\code|.
    \]
    
    We show that if there is an unlabeled edge $(v_i, v_j)$, then Eve has an attack, where Bob guesses incorrectly on at least one of $v_i$ and $v_j$. We describe an attack both in the case of errors and erasures.
    
    \begin{itemize}
    \item {\bf Errors:} 
        Regardless of whether Alice has $v_i$ or $v_j$ (which Eve may not know), Eve corrupts the string bit by bit to look like what Alice would send if she had $v_i$, and tracks how much corruption she would have needed to use if Alice had $v_j$ (if Alice actually has $v_j$ she uses this much corruption, and otherwise uses $0$ corruption). Once this quantity reaches $\left( \frac14+\frac12\delta \right)|\code|$, she switches to corrupting bit by bit so that Bob receives what Alice would send if she had $v_j$. Then, regardless of whether Alice had $v_i$ or $v_j$, Bob received the same string while sending the same feedback $f$.
    
        We need to show that the total corruption in either case is at most $\left( \frac14+\frac12 \delta \right)|\code|$. If Alice has $v_j$, then by the definition of the attack, Eve corrupted at most $\left( \frac14+\frac12\delta \right)|\code|$ bits. If Alice had $v_i$, then the number of bits corrupted was at most 
        \[
            \Delta \left( \code(v_i;f), \code(v_j;f) \right) - \left( \frac14+\frac12\delta \right)|\code| \leq \left( \frac14 + \frac12\delta \right)|\code|.
        \]
        Then, whether Alice actually had $v_i$ or $v_j$, the adversary needed only corrupt $\le \left( \frac14 - \frac12 \delta \right) |\code|$ of Alice's communication, so that Bob cannot distinguish whether Alice had $v_i$ or $v_j$.
    
    \item {\bf Erasures:}
        In this case, regardless of whether Alice has $v_i$ or $v_j$, Eve erases a bit if what Alice would send if she has $v_i$ differs what she would send with $v_j$. Then, regardless of whether Alice had $v_i$ or $v_j$, Bob received the same string, sending the same feedback $f$.
    
        The number of bits erased in either case is the same: specifically, it is 
        \[
            \Delta \left( \code(v_i;f), \code(v_j;f) \right) < \left( \frac12+\delta \right)|\code|
        \]
        since $(v_i, v_j)$ has no graph label. Then, knowing that the adversary is limited to erasing $\left( \frac12 + \delta \right) |\code|$ of Alice's communication, Bob cannot tell if Alice had input $v_i$ or $v_j$.
    \end{itemize}
    
    Thus, in order for a feedback ECC to have resilience to $\ge \frac14 + \delta$ errors or $\ge \frac12 + \delta$ erasures, every edge of the graph $G$ must have a label. 
    
    \begin{lemma}
        There is no $K_{\lceil 10/\delta\rceil}$ in $G$ with the same label $f$ on all edges.
    \end{lemma}
    
    \begin{proof}
        If there was one, then for all pairs $(v_i, v_j)$ in that subgraph, we would have 
        \[
            \Delta \left( \code(v_i;f), \code(v_j;f) \right) \geq \left( \frac12+\delta \right)|\code|
        \]
        However, this is impossible by the Plotkin bound \cite{Plotkin60}. 
    \end{proof}
    
    Therefore, we can view the labeling of $G$ as a coloring of the edges with $\ell = 2^{\zeta(\code)}$ colors avoiding any monochromatic $K_{\lceil 10/\delta\rceil}$. By a well-known bound on multicolor Ramsey numbers (Theorem~\ref{thm:ramsey}), the number of labels $\ell$ necessary for a graph on $K = 2^k$ vertices must satisfy 
    \begin{align*}
        &~\ell^{{\lceil 10/\delta\rceil}\ell}>K \\
        \implies &~ \ell = \Omega_\delta\left(\frac{\log{K}}{\log\log{K}}\right)
    \end{align*}
    Therefore, the number of possibilities for $f$ must be $\Omega_\delta\left(\frac{\log{K}}{\log\log{K}}\right)$, so the number of feedback bits transmitted by Bob is at least $\log{\Omega_\delta\left(\frac{\log{K}}{\log\log{K}}\right)}=\Omega_\delta(\log{k})$.
\end{proof}

\subsection{Lower Bound on Number of Rounds of Feedback} \label{sec:lb-rounds}

Next we show that for any fixed number of rounds of feedback $r$, the error (resp. erasure) resilience is capped strictly below $\frac13$ (resp. $1$). In other words, to get arbitrarily close to $\frac13$ error resilience or $1$ erasure resilience, the number of rounds of feedback needed must grow arbitrarily large.

First, we show this for errors.

\begin{theorem}
    For any $r\in \bbN$, for any feedback ECC family $\{\code\}_{k\geq 0}$ where Bob speaks in $r=\rho(\code)$ rounds, there exists $\delta(r)<\frac13$ such that the error resilience of $\code$ cannot exceed $\delta$.
\end{theorem}

\begin{proof}
    We prove the statement by induction on $r$. For $r=0$, we can let $\delta(0)=\frac14$ because the parties cannot exceed $\frac14$ error resilience by the Plotkin bound. Now, assuming that $\delta(r)<\frac13$, we show $\delta(r+1)<\frac13$. Let the number of bits Alice sends be $N$, and say that Bob sends his first round of feedback after Alice sends $b$ bits.
    
    \begin{caseof}
    \case{$b<\left(\frac13-\delta(r) \right)N$.}{
        Eve corrupts all of the first $b$ bits Alice sends to $0$. Then, Bob's first round of feedback cannot be useful to the parties since Alice already knows that Bob received $0^b$ so far. Thus, by the induction hypothesis, the adversary has an attack where they corrupt the remaining $\delta(r)(N-b)$ bits of the protocol and the first $b$ bits, for a total of
        \begin{align*}
            &~ b+\delta(r)(N-b) \\
            \leq &~\left(\frac13-\delta(r)\right)\cdot N + \delta(r)\left(\frac23+\delta(r) \right)\cdot N \\
            = &~\frac13N-\frac13\delta(r) N + \delta(r)^2N \\
            = &~\left(\frac13-\frac13\delta(r) + \delta(r)^2\right)\cdot N
        \end{align*}
        bits of corruption, and so we can take $\delta(r+1)=\left(\frac13-\frac13\delta(r) + \delta(r)^2\right)<\frac13$.
    }
    \case{$b\geq \left(\frac13-\delta(r) \right)N$.}{
        There exists sufficiently large $k$ such that there are $3$ possible inputs $x_1,x_2,x_3$ with the following property. If Alice's first $b$ bits for each input is denoted $\ECC_b(x_1), \ECC_b(x_2), \ECC_b(x_3)$, then since the error resilience threshold for list-of-$2$ equals $1/4$~\cite{Blinovski86,AlonBP19}, there exists $m\in \{0,1\}^b$ such that
        \[
            \Delta(\ECC_b(x_1),m), \Delta(\ECC_b(x_2),m), \Delta(\ECC_b(x_3),m)
            \leq \left(\frac14+0.05\right)\cdot b = \frac3{10}b.
        \]
        Then, when Alice has any one of these $3$ inputs, Eve corrupts the first $b$ bits to look like $m$, using $\frac3{10}b$ bits of corruption. Then, there exists a strategy for her to cause Alice to output the wrong answer in at least one of the three cases corrupting only $\frac13$ of the remaining bits plus $1$ by Theorem~\ref{thm:1/3-unlimited} in Appendix~\ref{sec:appendix}. In total, this uses
        \begin{align*}
            &~\frac3{10} b + \frac13(N-b) +1 \\
            \leq &~ \frac1{10}N-\frac3{10}\delta(r) N + \frac29N+\frac13\delta(r)N +1\\
            =&~ \frac{29}{90}N+\frac1{30}\delta(r) N +1\\
            =&~ \left(\frac{29}{90}+\frac1{30}\delta(r)\right) N +1
        \end{align*}
        bits of corruption, and so we can take $\delta(r+1)=\left(\frac{29}{90}+\frac1{30}\delta(r)\right) + \frac1N<\frac13$ as long as the total length of the protocol $N$ is large enough, and since $N\geq k$ this holds if we take large enough $k$. \qedhere
    }
    \end{caseof}
\end{proof}

The proof for erasures is very similar, if a bit simpler.

\begin{theorem}
    For any $r\in \bbN$, for any feedback ECC family $\{\code\}_{k\geq 0}$ where Bob speaks in $r=\rho(\code)$ rounds, there exists $\delta(r)<1$ such that the erasure resilience of $\code$ cannot exceed $\delta$.
\end{theorem}

\begin{proof}
    We prove the statement by induction on $r$. For $r=0$, we can let $\delta(0)=\frac12$ because the parties cannot exceed $\frac12$ erasure resilience by the Plotkin bound. Now, assuming that $\delta(r)<1$, we show $\delta(r+1)<1$. Let the number of bits Alice sends be $N$, and say that Bob sends his first round of feedback after Alice sends $b$ bits.
    
    \begin{caseof}
    \case{$b<\left(1-\delta(r) \right)N$.}{
        Eve erases all of the first $b$ bits. Then, Bob's first round of feedback cannot be useful to the parties. Thus, by the induction hypothesis, Eve has an attack where she erases the remaining $\delta(r)(N-b)$ bits of the protocol and the first $b$ bits, for a total of
        \begin{align*}
            &~ b+\delta(r)(N-b) \\
            \leq &~\left(1-\delta(r)\right)\cdot N + \delta(r)^2\cdot N \\
            = &~\left(1-\delta(r) + \delta(r)^2\right)\cdot N
        \end{align*}
        bits of corruption, and so we can take $\delta(r+1)=\left(1-\delta(r) + \delta(r)^2\right)<1$.
    }
    \case{$b\geq \left(1-\delta(r) \right)N$.}{
        We define $\ECC_b(x)$ to be the first $b$ bits that Alice sends on input $x$. There exists sufficiently large $k$ such that by the Plotkin bound~\cite{Plotkin60}, there are two possible inputs $x_1,x_2$ such that
        \[
            \Delta(\ECC_b(x_1),\ECC_b(x_2))
            \leq \left(\frac12+0.1\right)\cdot b = \frac3{5}b.
        \]
        Then, when Alice has any one of these two inputs, Eve erases the overlapping bits, using $\frac3{10}b$ bits of corruption. Then, she erases the remainder of the protocol. In total, this uses
        \begin{align*}
            &~\frac35 b + (N-b) \\
            \leq &~ \left(\frac35+\frac25\delta(r)\right) N
        \end{align*}
        bits of corruption, and so we can take $\delta(r+1)=\left(\frac35+\frac25\delta(r)\right)<1$. \qedhere
    }
    \end{caseof}
    
\end{proof}


\bibliographystyle{alpha}
\bibliography{refs}

\appendix
\section{Prior Work on Feedback ECC}\label{sec:appendix}

In this appendix, we describe two aspects of prior work on feedback ECC's that are not directly addressed in our paper, but useful to understand for a complete picture on feedback ECC's with optimal error resilience.

\subsection{Rewind-If-Error} \label{sec:appendix-rewind-if} 

Most existing schemes for achieving $\frac13$ error resilience for (unlimited) noiseless feedback are based on one of two paradigms. The first is the coin game approach \cite{Berlekamp64, SpencerW92} which we've described in detail in our paper. The second is \emph{rewind-if-error}. Though we do not use it in our paper, it is very important to the story of feedback ECC's. For example, many interactive coding protocols \cite{Schulman92, BrakerskiK12, Haeupler14, EfremenkoGH16, GuptaZ22a, GuptaZ22c} rely on the setup of rewind-if-error. Moreover, the limited feedback ECC of \cite{HaeuplerKV15} uses rewind-if-error. To illustrate the idea, we describe a rewind-if-error based feedback ECC for sending a message $\{0,1\}^k$ that achieves $\frac13$ error resilience with full noiseless feedback in $O(k)$ rounds.

\protocol{Rewind-If-Error Feedback ECC}{rewind-if}{
    Fix $\epsilon>0$. Alice has a message $x\in \{0,1\}^k$, which she pads with $1$'s to a message $X\in\{0,1\}^{k/\epsilon}$. The protocol lasts $k/\epsilon$ rounds. Bob tracks a guess $\hat{x}$ with $|\hat{x}|\leq k/\epsilon$ for Alice's input. At the beginning of the protocol, $\hat{x} = \emptyset$ is the empty string. The idea is that Alice will tell him the value of $X$ bit by bit, correcting errors as they appear. In every round, Bob sends a length $\leq k/\epsilon$ message equaling his current guess $\hat{x}$, and Alice sends Bob a length $3$ message. Alice and Bob act as follows:
    
    \begin{itemize}
        \item {\bf Alice's message:} Alice receives Bob's current guess $\hat{x}$. If $\hat{x}$ is a prefix of $X$, then Alice sends Bob the next bit of $X$ (specifically $X[|\hat{x}|+1]$) and otherwise she tells him to rewind. She embeds this instruction into a $3$-bit binary message as follows:
        
        She sends
        \[\begin{cases}
            100 & \text{if next bit is $0$}\\
            010 & \text{if next bit is $1$}\\
            001 & \text{if $\hat{x}$ is not a prefix of $X$}.\\
        \end{cases}\]
        \item {\bf Bob's update:} If Bob receives $100$, he appends $0$ to $\hat{x}$. If he receives $010$, he appends $1$ to $\hat{x}$. If he receives $001$, he removes the last bit from $\hat{x}$. If he receives anything else, he makes no updates to $T$.
    \end{itemize}
    At the end of the protocol, Bob outputs the first $k$ bits of his transcript guess $\hat{x}[1:k]$.
}

\begin{claim}
    Protocol~\ref{prot:rewind-if} is resilient to $\frac13-\epsilon$ corruption of Alice's bits.
\end{claim}

\begin{proof}
    Note that for any incorrect instruction Bob follows, one correct instruction undoes it. Thus, at the end of the protocol, Bob makes the correct guess if he received a correct instruction at least $k$ times more than he received an incorrect instruction. In order to give Bob an incorrect instruction, Eve needs to corrupt $2/3$ bits in Alice's message, and if she only corrupts $\frac13$ bits (resulting in a \emph{no-progress instruction}), then Bob makes no update to $\hat{x}$. Let the number of incorrect instructions if $I$ and the number of no-progress instructions is $N$. Then the number of correct instructions is $C=k/\epsilon-I-N$. Then,
    \begin{align*}
        &~ \frac23 I + \frac13 N \leq \left(\frac13 - \epsilon\right)k/\epsilon \\
        \implies &~ 2I+N\leq k/\epsilon-3k \\
        \implies &~ I+3k\leq C.
    \end{align*}
    As such, the number of correct instructions is at least $k$ (in fact $3k$) more than the number of incorrect instructions as desired.
\end{proof}

\subsection{Optimality of $\frac13$ Error Resilience}\label{sec:appendix-1/3}

Secondly, we show an upper bound (impossibility result) that no feedback ECC (with any amount of noiseless feedback) can achieve an error resilience greater than $\frac13$. In fact, we show the stronger statement that this holds as long as the there are at least three possible inputs Alice may hold. This was shown as early as \cite{Berlekamp64} when the concept was introduced, but we restate a proof here.

\begin{theorem}\label{thm:1/3-unlimited}
    No binary protocol $\pi$ of any length where Alice is trying to communicate one of three messages $a,b,c$ to Bob is resilient to more than $|\pi|/3+1$ bits of adversarial error. In every round of $\pi$, Alice sends exactly one bit, and is told whether the adversary corrupted her message, and the adversary may flip up to $\frac13|\pi|+1$ bits.
\end{theorem}

\begin{proof}[Proof (adapted from \cite{Gelles-survey})]
    Alice holds one of three possible inputs $a, b, c$. Consider the following attack, defined on the rounds $t\in [|\pi|]$. Denote the bit Alice sends in round $t$ by $b_a(t), b_b(t), b_c(1)$, respectively, based on Bob's past feedback (so in particular depending on the adversary's previous corruptions).  In every round $t \le R$ (we set $R$ shortly), Eve's attack changes the bit Alice sends to $\maj(b_a(t), b_b(t), b_c(t))$. Then, for at least two inputs, the channel does nothing, and for the third input, the channel may need to flip the bit. 
     
    Denote by $N_a(t)$ the noise the preceding attack causes through round $t$ if Alice holds the input $a$ (similarly define $N_b(t), N_c(t)$). Set $R$ to be the minimal round after which the second-largest value out of $N_a(R), N_b(R), N_c(R))$ is equal to $\lfloor|\pi|/3\rfloor$ and the largest value is strictly greater. Without loss of generality, assume $c$ maximizes this value at round $R$ and that $b$ is the second-largest value, that is, $N_a(R) \le \lfloor|\pi|/3\rfloor = N_b(R) < N_c(R)$. Also note that since the attack makes a single corruption in at most one input every round, we have $R \geq N_a(R) + N_b(R) + N_c(R)$. Finally, note that up to round $R$, Bob sees exactly the same view whether Alice holds $a, b, c$. From this point on, we don’t care about the input $c$, as we will show the attack succeeds on inputs $a$ and $b$.

    From round $R+1$ until the end of the protocol (round $|\pi|$), if Alice holds $a$, Eve changes all Alice’s transmissions to be what Alice would have sent given that she had the input $b$. If Alice holds $b$, the channel does nothing. In Bob’s view, exactly the same bits are received between round $R$ and the end of the protocol (and therefore, throughout the entire protocol) whether Alice holds $a$ or $b$.
    
    We are left to show that the total noise rate is at most $\frac13$. If Alice holds $b$, then the attack stops at round $R$, when $N_b(R) = \lfloor|\pi|/3\rfloor$, as needed. When Alice holds $a$, the channel corrupts $N_a(R)$ bits until round $R$ and at most $(|\pi| - R)$ bits afterwords. Recall that R $\geq N_a(R) + N_b(R) + N_c(R)$ and that $N_b(R)+N_c(R) \geq 2\lfloor|\pi|/3\rfloor+1$; thus the attack makes at most 
    \[
        N_a(R) + (|\pi| - R) \leq N_a(R) + |\pi| - N_a(R) - 2\lfloor|\pi|/3\rfloor-1 \leq |\pi|/3 + 1
    \]
corruptions, as needed.
\end{proof}

\end{document}